\documentclass[12pt]{article}
\usepackage[nonatbib,preprint]{neurips_2020}
\usepackage[numbers,compress]{natbib}
\usepackage{booktabs}       
\usepackage{amsfonts}       
\usepackage{nicefrac}       
\usepackage{graphicx}
\usepackage[dvipsnames]{xcolor}
\usepackage{enumitem}
\usepackage{rotating}
\usepackage{array}
\usepackage{mathtools}

\usepackage{wrapfig}
\usepackage{url,lineno,microtype,subcaption}
\usepackage[onehalfspacing]{setspace}
\usepackage{blindtext}
\usepackage[ruled,linesnumbered]{algorithm2e}
\usepackage{cite}
\usepackage{bm}
\newcommand{\ind}{\perp\!\!\!\!\perp}
\newcommand{\e}{\mathbb{E}}

\usepackage{amsmath,amsthm,amssymb}
\theoremstyle{definition}

\newtheorem{corollary}{Corollary}[section]
\newcommand\numberthis{\addtocounter{equation}{1}\tag{\theequation}}
\newtheorem{theorem}{Theorem}[section]
\newtheorem{lemma}{Lemma}[section]
\newtheorem{example}{Example}[section]

\newtheorem{remark}{Remark}[section]
\makeatletter
 \def\@textbottom{\vskip \z@ \@plus 1pt}
 \let\@texttop\relax
\makeatother

\newtheorem{definition}{Definition}
\usepackage[final]{hyperref} 
\hypersetup{
	colorlinks=true,       
	linkcolor=blue,        
	citecolor=blue,        
	filecolor=magenta,     
	urlcolor=blue         
}

\title{Consistent Causal Inference from Time Series with PC Algorithm and its Time-Aware Extension}
\author{
  Rahul Biswas\\ 
  Department of Statistics\\
  University of Washington\\
  Seattle, WA, 98195\\
  \texttt{rbiswas1@uw.edu}\\
  \And
  Somabha Mukherjee \\
  Department of Statistics and Data Science\\
  National University of Singapore\\
  Singapore, 117546\\
  \texttt{somabha@nus.edu.sg}\\
}
\begin{document}
\maketitle

\begin{abstract}
The estimator of a causal directed acyclic graph (DAG) with the PC algorithm is known to be consistent based on independent and identically distributed samples. In this paper, we consider the scenario when the multivariate samples are identically distributed but not independent. A common example is a stationary multivariate time series. We show that under a standard set of assumptions on the underlying time series involving $\rho$-mixing, the PC algorithm is consistent in this dependent sample scenario. Further, we show that for the popular time series models such as vector auto-regressive moving average and linear processes, consistency of the PC algorithm holds. We also prove the consistency for the Time-Aware PC algorithm, a recent adaptation of the PC algorithm for the time series scenario. Our findings are supported by simulations and benchmark real data analyses provided towards the end of the paper.
\end{abstract}
\section{Introduction}
Directed probabilistic graphical models are a popular tool to find causal relations between variables from their observational data. The causal relations are often represented by directed acyclic graphs (DAGs), which have its nodes as the variables and edges encoding information on conditional dependence between the variables. The conditional dependencies are specified by the directed Markov property \citep{lauritzen1996graphical}.

Estimation of the DAG from independent and identically distributed (i.i.d.) data can be conducted by score-based, constraint-based or hybrid methods. Score-based methods search over the space of all possible DAGs and maximize a goodness-of-fit score such as GES (greedy equivalent search) \citep{chickering2002optimal} and GIES (greedy interventional equivalent search, for data with experimental interventions) \citep{hauser2012characterization}. Constraint-based methods first use conditional dependence tests to find an undirected skeleton graph and then use the conditional dependence information in the skeleton to partially direct the edges,such as in the PC algorithm (acyclic graph, no latent confounders, no selection bias) \citep{spirtes2000causation} and FCI algorithm (acyclic graph, latent confounders, selection bias) \citep{spirtes1999algorithm}. Finally, hybrid methods combine score and constraint-based approaches for example the max-min hill climbing method \citep{tsamardinos2006max}.

Among the different approaches, PC algorithm remains one of the most popular methods for causal inference \citep{spirtes2000causation,glymour2019review}. The PC algorithm estimates the completed partially directed acyclic graph (CPDAG) which represents the class of all DAGs which are Markov Equivalent and thereby indistinguishable from observational data. In terms of computational time complexity, the PC algorithm is exponential (as a function of the number of nodes) in the worst case, but if the true DAG is sparse, which is often a reasonable assumption, it reduces to polynomial complexity. The question of asymptotic guarantees in estimation of the CPDAG by the PC algorithm has been studied in \citep{spirtes2000causation,robins2003uniform,kalisch2007estimating}. In particular, assuming only faithfulness, it has been shown that pointwise consistency of the PC algorithm can be achieved based on i.i.d. observations \citep{spirtes2000causation}. Furthermore, high dimensional uniform consistency of the PC algorithm based on i.i.d. samples has been shown to be achieved with a further Gaussian distribution assumption \citep{kalisch2007estimating}, which popularized the use of PC algorithm in high dimensional settings.

Despite its asymptotic guarantees in estimation from i.i.d. data, in the dependent sampling i.e. time series scenario, the consistency of the PC algorithm is still unclear. Traditionally methods like Granger Causality and Transfer Entropy have been used for answering causal queries \citep{dhamala2008analyzing,abdul2014quantifying} in time series, and recently the usage of graphical models, and in particular the PC algorithm, has drawn attention \citep{biswas2022statistical1,smith2011network,dahlhaus2003causality,ebert2012causal,runge2019detecting}. Recent approaches such as the Time-Aware PC (TPC) \citep{biswas2022statistical2} algorithm uses the PC algorithm after transformations on the time series and has been shown to have effective performance across simulated and benchmarking real datasets. To obtain clarity on the consistency of such methods as the TPC algorithm which builds upon the PC algorithm in time series, it is first essential to have clarity on the consistency of the PC algorithm in a time series scenario.


In this paper we prove consistency of the PC algorithm in the dependent sampling scenario, where the number of variables $p$ is fixed, and there are $n$ dependent identically distributed samples. A common example is a stationary multivariate time series. In particular, we show that under a standard set of assumptions, the PC algorithm is consistent in the dependent sampling scenario, for two well-known tests for conditional dependence. Specifically, we show that the consistency of the PC algorithm holds for a multivariate time series, under conditions of $\rho$-mixing (see \ref{eq:rhomixing}) with `sufficiently fast' convergence of the maximal correlation coefficients and faithfulness of the random vector at each time. This enables us to establish the consistency of PC algorithm in common time series models such as vector auto-regressive moving average models and linear processes. We also show that the Time-Aware PC algorithm achieves consistency under similar assumptions but with faithfulness with respect to a DAG across time points. The latter is more reasonable in a time series context with interactions between variables across different times. In the following section, we describe the conditions for causal inference that will provide the framework for the methods studied in the paper.


\section{Causal Structure Learning}

We begin by describing the general framework for causal structure learning by directed graphical models. A graph $G=(V,E)$ consists of a set of nodes $V=\{1,\ldots,p\}$ for fixed $p$, and a set of edges $E\subseteq V\times V$. In our setting the set of nodes corresponds to the components of a $p$-variate random vector $\bm{X}\in \mathbb{R}^p$. An edge $(u,v) \in E$ is called directed and denoted by $u\rightarrow v$ if $(u,v)\in E$ but $(v,u)\not \in E$. An edge between $u$ and $v$ is called undirected and denoted by $u-v$ if both $(u,v)\in E$ and $(v,u)\in E$. A Directed Acyclic Graph (DAG) is a graph with all edges directed and devoid of directed cycles. 

We will now define some graphical preliminaries which will be used in this section (see \citep{drton2016structure}). In a DAG $G$, we call a pair of nodes $v$ and $w$ adjacent if $v\rightarrow w \in E$ or $w\rightarrow v \in E$. A path is a sequence of distinct nodes in which successive nodes are adjacent. $v_0$ and $v_k$ are called endpoints of the path $\pi = (v_0, v_1 \ldots, v_k)$. For a  non-endpoint $v_i$, if $v_{i-1}\rightarrow v_i \leftarrow v_{i+1}$ is a subpath of $\pi$ then $v_i$ is called a collider on $\pi$. Else, $v_i$ is called a non-collider on $\pi$. $v_0$ is called an ancestor of $v_k$ (equivalently, $v_k$ is called a descendant of $v_0$) if every edge on $\pi$ is of the form $v_{i-1}\rightarrow v_i$. A common convention is to refer to $v$ as an ancestor and descendant of itself. We denote the set of ancestors and descendants of $v$ in $G$ by $an_G(v)$ and $de_G(v)$ respectively, and define $an_G(C):= \cup_{v\in C}an_G(v), ~~de_G(C):= \cup_{v\in C}de_G(v)$. We refer to the skeleton of a DAG $G$ as the undirected graph obtained by replacing directed edges in $G$ with undirected edges. An ordered triplet of nodes $(u,v,w)$ with $u$ and $w$ not adjacent in $G$ and $G$ containing the directed edges $u\rightarrow v$ and $w\rightarrow v$ is called a v-structure in $G$. The following definition of d-separation is taken from \citep{drton2016structure}.

\begin{definition}[d-separation]
Two nodes $v$ and $w$ in a DAG $G$ are \emph{d-connected} given $C\subset V\setminus \{v,w\}$ if $G$ contains a path $\pi$ with endpoints $v$ and $w$ such that (i) the colliders on $\pi$ are in $an_G(C)$, and (ii) no non-collider on $\pi$ is in $C$. Generalizing to sets, two disjoint sets $A, B \subset V$ are \emph{d-connected} given $C\subset V\setminus (A\cup B)$ if there are two nodes $v\in A$ and $w\in B$ that are d-connected given $C$. If this is not the case, then $C$ \emph{d-separates} $A$ and $B$.
\end{definition}

\subsection{Conditions for Causal Inference}\label{sec:causalconds}
Spirtes and Pearl pioneered the use of DAGs in causal inference using the condition of \emph{causal sufficiency}, i.e. absence of hidden (or latent) variables \citep{spirtes2000causation}, and two conditions that relate the DAG and probability distributions, namely \emph{directed Markov property} and \emph{faithfulness}.

\begin{enumerate}
\item \textbf{Directed Markov Property.} Let $(X_v:v\in V)\sim P$. $P$ is said to satisfy the Directed Markov Property with respect to $G$ if, for all $A,B,C\subset V$,
\begin{equation}\label{eq:dmp}
C \text{ d-separates } A \text{ and } B \Rightarrow X_A\ind X_B\vert X_C
\end{equation}
where $\ind$ represents independence of two random variables.
\item \textbf{Faithfulness.} Let $P$ satisfy the Directed Markov Property. $P$ is said to be faithful with respect to $G$ if the converse of \eqref{eq:dmp} holds, consequently i.e., for all $A,B,C\subset V$,
\begin{equation}\label{eq:ftflns}
X_A\ind X_B\vert X_C \Leftrightarrow C \text{ d-separates } A \text{ and } B 
\end{equation}
\end{enumerate}

It is noteworthy that faithfulness constrains the class of probability distributions. See Chapter 3.5.2 in \citep{spirtes2000causation} for an example of a non-faithful distribution. In this paper, we mostly limit ourselves to the multivariate Gaussian family, where non-faithful distributions form a Lebesgue null set in the space of distributions associated with a DAG G (see \citep{meek1995strong}).

Most algorithms for estimating a DAG $G$ under the Conditions 1-3 cannot distinguish between two DAGS which are \emph{Markov equivalent} (two DAGs $G_1$ and $G_2$ are said to be Markov equivalent when the set of distributions that are faithful with respect to $G_1$ is the same as the set of distributions that are faithful with respect to $G_2$ \citep{chickering2002learning}). Instead, one can identify the Markov equivalence class of DAGs. The Markov equivalence classes of DAGs can also be characterized by the following criterion: Two DAGs $G_1$ and $G_2$ are Markov equivalent if and only if they have the same skeleton and v-structures \citep{verma2022equivalence}. It is commonplace for a Markov equivalence class to be conveniently represented by a Completed Partially Directed Acyclic Graph (CPDAG) which has a directed edge $v \rightarrow w$ if $v\rightarrow w$ is present in all the DAGs belonging to the Markov equivalence class of the DAG and has an undirected edge between $v$ and $w$ if both $v\rightarrow w$ and $w\rightarrow v$ are present among the DAGs. A CPDAG represents a Markov equivalence class uniquely, since two CPDAGs are identical if and only if they represent the same Markov equivalence class \citep{chickering2002learning}.

Therefore, instead of estimating the DAG satisfying the conditions of causal inference, the main goal is to estimate the CPDAG. We focus on the PC algorithm in this paper which is a popular method for this estimation task \citep{spirtes2000causation, pearl2009causality}. Although the main goal is to estimate the CPDAG, PC algorithm divides the task into two parts 1) Estimation of the skeleton of the true DAG, and 2) Orientation of edges of the skeleton to obtain the CPDAG estimate. All inference from data takes place in the first part. In the following section we describe the steps of the PC algorithm which will be used later to establish conditions for its consistency in dependent data scenario.

\subsection{Estimation of the CPDAG: The PC algorithm}\label{subsec:pcpop}
The PC algorithm is a popular method to estimate the CPDAG from observed data. It has been shown to be pointwise consistent based on independent and identically distributed data using a consistent test for conditional dependence, and uniformly consistent\footnote{See \citep{robins2003uniform} for definition of uniform consistency in this setting.} when additionally the data distribution is Gaussian \citep{kalisch2007estimating}. 

The population version of the algorithm has two parts: 1) To find the skeleton $G_{skel}$ of the true DAG $G$ by representing the variables as nodes of an empty DAG and putting an undirected edge between each pair of nodes if they are independent or conditionally independent given any other variable(s). 2) The skeleton is then converted into the CPDAG $G_{CPDAG}$ by orienting the undirected edges using rules for orientation. In the sample version the independence and conditional independence statements are replaced by a statistical test based on the sample. The parts of the population version of the PC algorithm are outlined in Algorithms \ref{alg:pcpop} and \ref{alg:pcorient}. For a DAG $G$ such that the distribution $P$ is faithful to $G$, it is proved that the PC algorithm \ref{alg:pcpop} constructs the true skeleton of the DAG \citep{kalisch2007estimating}, and algorithm \ref{alg:pcorient} outputs the true CPDAG \citep{meek1995causal}. In the sample version only the conditional independence statements in finding the skeleton in Algorithm \ref{alg:pcpop} are replaced with a statistical test for conditional dependence based on data.
\begin{algorithm}[t!]
\SetKwInOut{Input}{Input}
\SetKwInOut{Output}{Output}
\Input{Node set V, Conditional Independence Information}
\Output{Skeleton $G_{skel}$, separation sets $S$}

Start with a complete undirected graph $C$ on the vertex set V.

Set $\ell = -1; \quad G_{skel}=C$

\Repeat{for each $i$,$j$ adjacent: $|adj(G_{skel},i) \setminus
    \{j\}| \le \ell$.}{

    $\ell=\ell +1$

    \Repeat{all $i,j$ adjacent such that $|adj(G_{skel},i) 
    \setminus \{j\}| \geq \ell$ and $\bm{k} \subseteq adj(G_{skel},i) \setminus
    \{j\}$ with 
    $|\bm{k}|=\ell$ have been tested for conditional independence.}{

    Select a new ordered pair $i,j$ adjacent in $G_{skel}$ with $|adj(G_{skel},i) \setminus \{j\}| \geq \ell$ 

        \Repeat{edge $i,j$ is deleted or all $\bm k \subset  adj(G_{skel},i)\setminus \{j\}$ with $|\bm{k}|=\ell$ are selected}{
            Choose $\bm{k} \subseteq adj(G_{skel},i) \setminus \{j\}$ with $|\bm{k}|=\ell$

            \uIf{$i\ind j ~\vert~ \bm{k}$}{
                
                Delete edge $i-j$
                
                Denote this new graph by $G_{skel}$
                
                Save $\bm k$ in $S(i,j)$ and $S(j,i)$}
}
}
}

\caption{Population PC - Finding the skeleton graph}\label{alg:pcpop}
\end{algorithm}
\begin{algorithm}[t!]
\SetAlgoLined
\SetKwInOut{Input}{Input}
\SetKwInOut{Output}{Output}
\Input{Skeleton $G_{skel}$, separation sets $S$}
\Output{$G_{CPDAG}$}
\textbf{for all} $i,j$ non-adjacent\\ 
\quad \quad Orient $i-c-j$ whenever $i,j$ are not conditionally independent given $c$ into $i\rightarrow c\leftarrow j$.\\

\quad \quad Orient $i\rightarrow j-c$ with $i,j$ non-adjacent into $i\rightarrow j\rightarrow c$\\
\textbf{end for}\\

\textbf{for all} $i,j$ adjacent\\
\quad \quad Orient $i-j$ whenever $i\rightarrow c \rightarrow j$ into $i\rightarrow j$\\
\quad \quad Orient $i-j$ whenever $i-c\rightarrow j$ and $i-d\rightarrow j$ for some $c,d$ nonadjacent into $i\rightarrow j$.\\
\quad \quad Orient $i-j$ whenever $i-c\rightarrow d$ and $c\rightarrow d\rightarrow j$ for some $c,d$ nonadjacent into $i\rightarrow j$.\\
\textbf{end for}\\
\caption{Population PC - Directing the skeleton into a CPDAG}\label{alg:pcorient}
\end{algorithm}

\subsection{Conditional Dependence Test in the Gaussian Regime: Pearson's Partial Correlations}\label{sec:samplepc}
In the Gaussian scenario conditional dependence can be inferred from partial correlations.  Assume that $\mathbf{X}=(X_1,\ldots,X_p)$ is a $p$-dimensional Gaussian random vector, for some fixed integer $p$. For $i \neq j \in \{1,\ldots ,p\},\ \bm{k} \subseteq \{1,\ldots ,p\}
  \setminus 
\{i,j\}$, denote by $\rho_{i,j|\bm{k}}$ the 
partial correlation between $X_i$ and $X_j$ given
$\{X_r:\ r \in \bm{k}\}$. We can infer conditional dependence by inferring partial correlations due to the following elementary property of the multivariate Gaussian distribution (see Prop. 5.2 in \citep{lauritzen1996graphical}) that, 

\[\rho_{i,j|\bm{k}}=0\text{ if and only if }X_i\ind X_j ~\vert~ \{X_r~:\ r \in \bm{k}\}.\] 

\medskip
The sample partial correlation
$\hat{\rho}_{i,j|\bm{k}}$ can be 
calculated via regression or by using the following identity. Denote $k=\vert \bm k\vert$ and let without loss of generality $\{X_r;\ r \in \bm{k}\}$ be the last $k$ entries in $\bm{X}$. Let $\Sigma:= \text{cov}(\bm{X})$ with $\Sigma = \left(\begin{array}{cc}
    \Sigma_{11} & \Sigma_{12} \\
    \Sigma_{21} & \Sigma_{22}
\end{array}\right)$ where $\Sigma_{11}$ is of dimension $(p-k)\times (p-k)$, $\Sigma_{22}$ is of dimension $k \times k$, $\Sigma_{11.2}=\Sigma_{11}-\Sigma_{12}\Sigma_{22}^{-1}\Sigma_{21}$, and $\hat{\Sigma}$ and $\hat{\Sigma}_{11.2}$ be the sample versions of $\Sigma$ and $\Sigma_{11.2}$ based on the sample covariance matrix. Let $\bm e_1,\ldots,\bm e_p$ be the canonical basis vectors of $\mathbb{R}^p$. It follows from \citep{muirhead2009aspects} that,
\begin{align*}
\rho_{i,j\vert \bm{k}} &= \frac{\bm{e}_i^{\top}\Sigma_{11.2}\bm{e}_j}{\sqrt{(\bm{e}_i^{\top}\Sigma_{11.2}\bm{e}_i(\bm{e}_j^{\top}\Sigma_{11.2}\bm{e}_j))}},\\
\hat{\rho}_{i,j\vert \bm{k}} &= \frac{\bm{e}_i^{\top}\hat{\Sigma}_{11.2}\bm{e}_j}{\sqrt{(\bm{e}_i^{\top}\hat{\Sigma}_{11.2}\bm{e}_i(\bm{e}_j^{\top}\hat{\Sigma_{11.2}}\bm{e}_j))}}
\end{align*}
For testing whether a partial correlation is zero or not, we first apply Fisher's
z-transform
\begin{eqnarray}\label{ztrans}
Z(i,j|\bm{k}) = g(\hat{\rho}_{i,j\vert \bm k})=\frac{1}{2} \log \left (\frac{1 +
    \hat{\rho}_{i,j|\bm{k}}}{1 - \hat{\rho}_{i,j|\bm{k}}} \right).
\end{eqnarray}

Let also, $z(i,j\vert \bm{k}) = g(\rho_{i,j\vert \bm k})$. Note that $z(i,j\vert \bm{k})=0 \Leftrightarrow \rho(i,j\vert \bm{k})=0$, and hence, $z(i,j|\bm k) = 0 \Leftrightarrow X_i\ind X_j ~\vert~ \{X_r~:\ r \in \bm{k}\}$. We will show that $Z(i,j|\bm k)$ is a consistent estimator of the population parameter $z(i,j|\bm k)$, and hence, a consistent conditional dependence test can be constructed based on the statistic $Z(i,j|\bm k)$.
Using such a test for conditional dependence, the sample version of the PC algorithm (see \ref{alg:pcsample}) only replaces line 9 in the population PC in \ref{alg:pcpop} about conditional independence with the statistical test.

\subsection{Conditional Dependence Test in the Non-Gaussian Regime: The Hilbert Schmidt Criterion}\label{sec:samplepch}

When the data is non-Gaussian, zero partial correlations do not necessarily imply conditional independence. In such a situation, the Hilbert Schmidt criterion acts as a convenient test for conditional dependence. We describe this concept below.

Given $\mathbb{R}$-valued random variables $X,Y$ and the random vector $\bm Z$ we propose the use of the following statistic for testing the conditional dependence of $X,Y\vert \bm{Z}$ (see \citep{fukumizu2007kernel}):

\[
\hat{H}_{n}(X,Y\vert\bm{Z}) = Tr[R_{\overset{..}{Y}}R_{\overset{..}{X}}-2R_{\overset{..}{Y}}R_{\overset{..}{X}}R_{\bm Z} + R_{\overset{..}{Y}}R_{\bm Z}R_{\overset{..}{X}}R_{\bm Z}]
\]
where $G_X,G_Y,G_{\bm Z}$ are the centered Gram matrices with respect to a positive definite and integrable kernel $k$, that is, $G_{X,ij}=<k(\cdot,X_i)-\hat{m}_X^{(n)},k(\cdot,X_j)-\hat{m}_X^{(n)}>$ with $\hat{m}_X^{(n)}=\frac{1}{n}\sum_{i=1}^n k(\cdot,X_i)$, and $R_X=G_X(G_X+n\epsilon_n I_n)^{-1}, R_Y=G_Y(G_Y+n\epsilon_n I_n)^{-1},R_{\bm Z}=G_{\bm Z}(G_{\bm Z}+n\epsilon_n I_n)^{-1}$ and $\overset{..}{X}=(X,{\bm Z}), \overset{..}{Y}=(Y,{\bm Z})$. Under some regularity assumptions mentioned below, it follows from the proof of Theorem 5 in \citep{fukumizu2007kernel} that $\hat{H}_n(X,Y\vert \bm{Z})$ is a consistent estimator of $H(X,Y\vert \bm Z):= \|V_{\overset{..}{X}\overset{..}{Y}|{\bm Z}}\|^2$, where $$V_{\overset{..}{X}\overset{..}{Y}|{\bm Z}} := \Sigma_{\overset{..}{X}\overset{..}{X}}^{-1/2}(\Sigma_{\overset{..}{X}\overset{..}{Y}} - \Sigma_{\overset{..}{X}\bm Z} \Sigma_{\bm Z \bm Z}^{-1}\Sigma_{\bm Z \overset{..}{Y}}) \Sigma_{\overset{..}{Y} \overset{..}{Y}}^{-1/2}$$ and $\Sigma_{UV}$ denotes the covariance matrix of $U$ and $V$. It follows from \citep{fukumizu2007kernel} that $X\ind Y |\bm Z \Leftrightarrow H(X,Y|\bm Z)=0$. Hence, a consistent conditional dependence test of $X,Y | \bm Z$ can be constructed based on the statistic $H_n(X,Y| \bm Z)$. The advantage of this method is that unlike the Pearson partial correlation, it does not require Gaussianity of the data to decide conditional independence, and hence can be used in the PC algorithm if the underlying time series is non-Gaussian.

\begin{algorithm}[t]
\SetKwInOut{Input}{Input}
\SetKwInOut{Output}{Output}
\Input{Dataset $\{\bm{X}_i\}_{i=1}^n$}
\Output{CPDAG estimate $\hat{G}_{CPDAG}$}
Run the Population PC Algorithm \ref{alg:pcpop} to find the skeleton but replace in line 9 the statement about conditional independence of $i,j$ given $\bm{k}$ by the event of acceptance of conditional independence by a consistent test for conditional dependence.

Extend the skeleton to a CPDAG using Algorithm \ref{alg:pcorient}.
\caption{Sample PC algorithm}\label{alg:pcsample}
\end{algorithm}

\section{Consistency of the PC Algorithm from Dependent Samples}

While the sample PC algorithm as described in Algorithm \ref{alg:pcsample} is known to be consistent in estimating the true CPDAG from independent and identically distributed data provided one uses a consistent conditional dependence test \citep{kalisch2007estimating, fukumizu2007kernel}
, its consistency, to the best of our knowledge, is unknown when the data are not independent such as in a time series scenario.  In the simple case when the data is Gaussian, the analysis for i.i.d. data in \citep{kalisch2007estimating, hotelling1953new} involves calculation of the density function of sample partial correlations and thus the conclusions do not trivially follow in a dependent data setting. In this section, we establish some structure on the dependent time series data that ensures consistency of the PC algorithm. An important structure that we assume for the time series is strongly mixing, which assumes a form of weak dependence among the time series variables with increasing time gap between the variables. As a preparation we introduce a general consistency framework below.

\subsection{The General Consistency Framework}
The aim of this section is to show that under certain assumptions on the time series, the CPDAG estimate of the PC algorithm is consistent, that is, the probability of an error in estimation of the CPDAG converges to $0$. For this, it suffices to show that the probability of an error in testing conditional dependence converges to $0$, a fact that is guaranteed if one works with a consistent test for conditional dependence in the PC algorithm.  To be precise, suppose that we start with a measure $\mu(i,j|\bm k)$ of conditional dependence, which satisfies the property that
\[\mu(i,j|\bm k)=0\text{ if and only if }X_i\ind X_j ~\vert~ \{X_r~: \ r \in \bm{k}\}~,\]
and suppose that for each $i,j, \bm k$, we have a sequence $\hat{\mu}_n(i,j|\bm k)$ of consistent estimators of $\mu(i,j|\bm k)$. Then, the test of conditional dependence based on $\hat{\mu}_n(i,j|\bm k)$, which will reject conditional independence of $X_i$ and $X_j$ given $\{X_r\}_{r\in \bm k}$ if and only if $|\hat{\mu}_n(i,j|\bm k)| > \gamma$ for some thresholding parameter $\gamma$, will be consistent. Consequently, the error probability in estimating the CPDAG converges to $0$. The threshold $\gamma$ can be obtained using a bootstrap method, e.g. the stationary bootstrap for dependent samples constituting a stationary time series \citep{politis1994stationary}.

Below, we state the result that guarantees consistent estimation of the DAG skeleton.

\begin{theorem}\label{thm:skelconsistency}
Denote by $\hat{G}_{skel,n}$ the estimate of the skeleton graph from the PC algorithm with $\hat{\mu}_n$ as the test statistic for conditional dependence, and by $G_{skel}$ the true skeleton of the DAG $G$. Then,
\[
P(\hat{G}_{skel,n}=G_{skel})\rightarrow 1 \quad\text{ as }\quad n\rightarrow \infty.
\]
\end{theorem}
\begin{proof}
An error occurs in the sample PC algorithm if there is a pair of nodes $i,j$ and a conditioning set $\bm{k}\in K_{i,j}$ where an error event $E_{i,j\vert \bm{k}}$ occurs, where $E_{i,j\vert \bm{k}}$ denotes that ``an error occurred when testing partial correlation for zero at nodes $i,j$ with conditioning set $\bm{k}$", i.e.,

\[
E_{i,j\vert \bm{k}} = E_{i,j\vert \bm{k}}^{I} \cup  E_{i,j\vert \bm{k}}^{II},
\]

where
\[
E_{i,j\vert \bm{k}}^{I} := \{|\hat{\mu}_n(i,j|\bm k)| > \gamma \text{ and } \mu(i,j|\bm k) = 0\}\]
\[E_{i,j\vert \bm{k}}^{II} := \{ |\hat{\mu}_n(i,j|\bm k)| \le \gamma \text{ and } \mu(i,j|\bm k) \neq 0\}
\] denote the events of Type I error and Type II error respectively. Thus,

\begin{align*}
P(\text{an error occurs in the PC algorithm})&\leq P\left(\bigcup_{i,j,\bm{k}\in K_{i,j}}E_{i,j\vert \bm{k}}\right)\\
& \leq O(1)\sup_{i,j,\bm{k}\in K_{i,j}}P(E_{i,j\vert \bm{k}}) \numberthis\label{eqn:skelerror}\\    
\end{align*}
using that the cardinality of the set $\vert \{i,j,\bm{k}\in K_{i,j}\} \vert = 2^{p-2}p^2$. Then, for any $\gamma > 0$, we have:

\begin{align*}
    \sup_{i,j,\bm{k}\in K_{i,j}} P(E_{i,j\vert \bm{k}}^{I}) &= \sup_{i,j,\bm{k}\in K_{i,j}} P(\vert \hat{\mu}_n(i,j|\bm k) - \mu(i,j|\bm k)\vert > \gamma)
\end{align*}
It now follows from the consistency of $\hat{\mu}_n$ that
\[\sup_{i,j,\bm{k}\in K_{i,j}} P(E_{i,j\vert \bm{k}}^{I}) \rightarrow 0\numberthis\label{eqn:skeltype1err}\]

Next, we control the type II error probability. Towards this,  let $c=\inf \{\vert \mu(i,j|\bm k) \vert : \mu(i,j|\bm k)\neq 0\}>0$, and choose $\gamma = c/2$. Then,
	
\begin{align*}
    \sup_{i,j,\bm{k}\in K_{i,j}} P(E_{i,j\vert \bm{k}}^{II}) &= \sup_{i,j,\bm{k}\in K_{i,j}} P(\vert \hat{\mu}_n(i,j|\bm k)\vert \leq \gamma ~, ~\mu(i,j|\bm k)\neq 0)\\
&\leq \sup_{i,j,\bm{k}\in K_{i,j}} P(\vert \hat{\mu}_n(i,j|\bm k) - \mu(i,j|\bm k) \vert > c/2))
\end{align*}
Since the right hand side $\rightarrow 0$ as $n\rightarrow \infty$, it follows that
\[\sup_{i,j,\bm{k}\in K_{i,j}} P(E_{i,j\vert \bm{k}}^{II}) \rightarrow0\numberthis\label{eqn:skeltype2err}\]
Now by Eqns (\ref{eqn:skelerror})-(\ref{eqn:skeltype2err}) we get
\[
P(\text{an error occurs in the PC algorithm})\rightarrow 0
\]
This completes the proof.
\end{proof}

Since all inference is done while finding the skeleton, if this part is completed perfectly, that is, if there was no error while testing for conditional dependence, then $G_{CPDAG}$ will be estimated without error (See \citep{meek1995causal}). Therefore, we easily obtain: 
\begin{theorem}\label{thm:cpdagconsistency}
Denote by $\hat{G}_{CPDAG,n}$ the estimate from the entire PC algorithm and by $G_{CPDAG}$ the true $CPDAG$ from the DAG $G$. Then, 

\[
P(\hat{G}_{CPDAG,n}=G_{CPDAG})\rightarrow 1\quad \text{as}~ n\rightarrow \infty.
\]
\end{theorem}

\subsection{Consistent Tests for Conditional Dependence}
In the Gaussian regime (Section \ref{sec:samplepc}), one can take $\mu(i,j|\bm k) = z(i,j|\bm k)$ and $\hat{\mu}_n(i,j|\bm k) := Z(i,j|\bm k)$, and in the non-Gaussian regime (Section \ref{sec:samplepch}), one can take $\mu(i,j|\bm k) = H(X_{1i},X_{1j}|X_{1r} :r\in \bm k)$ and $\hat{\mu}_n(i,j|\bm k) := H_n(X_{1i},X_{1j}|X_{1r} :r\in \bm k)$, where $\bm{X}_t = (X_{t1},X_{t2},\ldots, X_{tp})$. All that we need to show in order to establish consistency of the PC algorithm in these two settings, is the consistency of $\hat{\mu}_n(i,j|\bm k)$ as an estimator of $\mu(i,j|\bm k)$. With this in mind, we will delineate the conditions on the time series under which the consistency of the sample versions of the conditional dependence tests holds. We need a few notations first.

Let $\{\bm{X}_t\}_{t=1}^\infty$ be a strictly stationary $\mathbb{R}^p$-valued time series, where $\bm{X}_t\sim P$ for all $t\ge 1$. Let $\bm{Y}_t=\bm{X}_t-\e\bm{X}_t$ and denote by $\rho_{ij}$ the population (Pearson's) correlation between $X_{1i}$ and $X_{1j}$  (equivalently between $Y_{1i}$ and $Y_{1j}$), and by $\hat{\rho}_{n;i,j}$ the sample correlation between $\bm{X}^{(i)}=(X_{1i},X_{2i},\ldots,X_{ni})$ and $\bm{X}^{(j)}$ (equivalently between $\bm{Y}^{(i)}$ and $\bm{Y}^{(j)}$) for $n$ samples.

We now introduce the concept of $\rho$-mixing. For fixed $i,j \in 1,\ldots, p$, let $\mathcal{F}_{a}^{b}$ be the $\sigma$-field of events generated by the random variables $\{X_{ti},X_{tj}:a\leq t \leq b\}$, and $L_2(\mathcal{F}_a^b)$ be the collection of all second-order random variables which are $\mathcal{F}_a^b$-measurable. The stationary process $\{X_{ti},X_{tj}:t= 1,2,\ldots\}$ is called $\rho$-mixing \citep{kolmogorov1960strong,bradley2005basic} if the maximal correlation coefficient:

\begin{equation}\label{eq:rhomixing}
\xi_{ij}(k):=\sup_{l\geq 1}\sup_{\substack{U\in L_2(\mathcal{F}_{1}^l)\\ V\in L_2(\mathcal{F}_{l+k}^{\infty})}} \frac{\vert \text{cov}(U,V) \vert}{\text{var}^{1/2}(U)\text{var}^{1/2}(V)} \rightarrow 0 \text{ as } k\rightarrow \infty.
\end{equation}

Let us assume the following conditions.
\begin{enumerate}
    \item[(A.1)] $\{X_{ti},X_{tj}: t=1,2,\ldots\}$ is $\rho$-mixing for all $i,j$, with maximal correlation coefficients $\xi_{ij}(k), k\ge 1$.
     \item[(A.2)] $\e X_{ti}^4 <\infty$ for all $t,i$ and $\sum_{k=1}^{\infty} \xi_{ij} (k) < \infty$ for $i,j\in 1,\dots,p$.
    \item[(A.3)] There exists a sequence of positive integers $s_n\rightarrow \infty$ and $s_n = o(n^{1/2})$ such that $n^{1/2}\xi_{ij}(s_n)\rightarrow 0$ as $n\rightarrow\infty$ for $i,j\in 1,\ldots,p$.
    \item[(A.4)] $P$ is faithful to a DAG $G$.
\end{enumerate}

\subsubsection{The Gaussian Regime}
The following lemma shows that under assumptions A.1-A.3, the sample correlation between every pair of variables converges to the corresponding population correlation.

\begin{lemma}\label{lem:rhoconsistent}
Under A.1-A.3, 
\[
\hat{\rho}_{n;i,j} \text{ converges to } \rho_{i,j} \text{ in probability.}
\]

\end{lemma}
\begin{proof}
Note that $\bm Y_1,\ldots,\bm Y_n$ are centered random variables, and under A.1-A.3 together with the Gaussianity assumption on the data, each of their entries have finite fourth moment. Furthermore, the maximal correlation coefficients are invariant to centering and scaling of the random variables. Therefore, assumptions A.1-A.3 hold for $\bm{Y}_1,\ldots,\bm{Y}_n$ which are also Gaussian. Hence, by Theorem 3 in \citep{masry2011estimation}, the lemma follows.
\end{proof}


Denote the population partial correlation between $X_{1i}$ and $X_{1j}$ $\vert \{X_{1r}:r\in \bm{k}\}$ for some $\bm{k}\subset \{1,\ldots,p\}\setminus \{i,j\}$ by $\rho_{i,j\vert \bm{k}}$, and let $k:=\vert \bm k \vert$. Similarly denote by $\hat{\rho}_{n;i,j|\bm{k}}$, the sample Partial Correlation between $\bm{X}^{(i)}$ and $\bm{X}^{(j)}$ $\vert \{\bm{X}^{(r)}:r\in \bm{k}\}$ which is also the partial correlation between $\bm{Y}^{(i)}$ and $\bm{Y}^{(j)}$ $\vert \{\bm{Y}^{(r)}:r\in \bm{k}\}$, for $\bm{k}\subset \{1,\ldots,p\setminus\{i,j\}\}$. 

We will now show that Lemma \ref{lem:rhoconsistent} can in fact be used to prove convergence of the pairwise sample partial correlations to the corresponding population correlations.



\begin{lemma}\label{lem:parcorrconsistent}
Assume (A.1)-(A.3). Then,

\[
\hat{\rho}_{n;i,j\vert \bm{k}} \text{ converges to } \rho_{i,j\vert \bm{k}} \text{ in probability.}
\]
\end{lemma}
\begin{proof}
Without loss of generality let $\{X_{1r}:r\in \bm{k}\}$ be the last $k$ entries of $\bm X_1$. We will define a function $f$ on the set of all non-singular $M\in \mathbb{R}^{p\times p}$. Let $M = \left(\begin{array}{cc}
    M_{11} & M_{12} \\
    M_{21} & M_{22}
\end{array}\right)$ where $M_{11}$ is of dimension $(p-k)\times (p-k)$ and $M_{22}$ is of dimension $k \times k$, $M_{11.2}=M_{11}-M_{12}M_{22}^{-1}M_{21}$ and $e_1,\ldots,e_p$ denote the canonical basis vectors of $\mathbb{R}^p$. Define $f(M) = \frac{e_{i}^{\top} M_{11.2} e_{j}}{\sqrt{(e_{i}^{\top} M_{11.2} e_{i})(e_{j}^{\top} M_{11.2} e_{j})}}$. Clearly $f$ is a continuous function. 
Let $\hat{\Sigma}$ denote the sample covariance matrix of $\bm{X}_1,\ldots, \bm{X}_n$. 
Therefore as seen in Section 5.3 of \citep{muirhead2009aspects},
\[
\rho_{i,j\vert \bm{k}} = f(\Sigma)
\]
and, 
\[
\hat{\rho}_{i,j\vert \bm{k}} = f(\hat{\Sigma})
\]
Also note that $\Sigma = ((\rho_{ij}\sigma_i\sigma_j))$ and $\hat{\Sigma} = ((\hat{\rho}_{ij}\hat{\sigma}_i\hat{\sigma}_j))$ where $\rho_{ij}$ is the population correlation of $X_{1i}$ and $X_{1j}$ and $\hat{\rho}_{ij}$ is the sample correlation of $\bm{X}^{(i)}$ and $\bm{X}^{(j)}$, $\sigma^2_i$ is the population variance of $X_{1i}$, and $\hat{\sigma}^2_i$ sample variance of $\bm{X}^{(i)}$, where $\bm{X}^{(i)}=(X_{1i},X_{2i},\ldots,X_{ni})$. Now, from Corollary 1 in \citep{masry2011estimation}, it follows that $\hat{\sigma_i}\xrightarrow{P}\sigma_i$ for all $i$ under (A.0)-(A.3). Furthermore, $\hat{\rho}_{ij}\xrightarrow{P}\rho_{ij}$ according to Lemma \ref{lem:rhoconsistent} under (A.0)-(A.3). These together imply that $\hat{\Sigma}\xrightarrow{P}\Sigma$. Since, $f$ is a continuous function, we also have $\hat{\rho}_{i,j\vert \bm{k}} = f(\hat{\Sigma})\xrightarrow{P}f(\Sigma) =\rho_{i,j\vert \bm{k}}$, which completes the proof.
\end{proof}
The PC algorithm tests partial correlations after the Z-transform $g(\rho) = 0.5 \log((1+\rho)/(1-\rho))$ (see Section \ref{sec:samplepc} for details). Denote $Z_{n;i,j\vert \bm{k}} = g(\hat{\rho}_{n;i,j\vert \bm{k}})$ and $z_{i,j\vert \bm{k}}=g(\rho_{i,j\vert\bm{k}})$.
\begin{lemma}\label{cor4p1}
Under Assumptions (A.1)-(A.3), we have, for all $i,j, \bm k \subseteq K_{i,j}$:
\[
Z_{n;i,j\vert \bm{k}} \text{ converges to } z_{i,j\vert \bm{k}} \text{ in probability.}
\]
\end{lemma}
\begin{proof}
Lemma \ref{cor4p1} follows trivially from Lemma \ref{lem:parcorrconsistent}, on observing that the function $g$ is continuous.
\end{proof}

Since $Z_{n;i,j|\bm k}$ is a consistent sequence of estimators of $z_{i,j|\bm k}$, consistency of the PC algorithm using $Z_{n;i,j|\bm k}$ as a statistic for testing conditional dependence in the Gaussian regime now follows from Theorems \ref{thm:skelconsistency} and \ref{thm:cpdagconsistency}.

\subsubsection{The Non-Gaussian Regime}

In this section, we will show that under assumptions A.1 - A.3, $H_n(X_{1i},X_{1j}|X_{1r}~:r\in \bm k)$ is a consistent estimator of $H(X_{1i},X_{1j}|X_{1r}~:r\in \bm k)$. For notational convenience, let us abbreviate $H(X_{1i},X_{1j}|X_{1r}~:r\in \bm k)$ by $H_{i,j|\bm k}$ and $H_n(X_{1i},X_{1j}|X_{1r}~:r\in \bm k)$ by $H_{n;i,j|\bm k}$.

\begin{lemma}\label{lem:rhoconsistenth}
Under A.1-A.3, if the regularization constant $\epsilon_n$ satisfies $n^{-1/3} \ll \epsilon_n \ll 1$, then 
\[
H_{n;i,j\vert \bm{k}} \text{ converges to } H_{i,j \vert \bm{k}} \text{ in probability.}
\]
\end{lemma}
\begin{proof}
The proof of Lemma \ref{lem:rhoconsistenth} essentially follows from the proof of Theorem 5 in \citep{fukumizu2007kernel}, modulo the fact that the samples are no longer independent. To begin with, it follows from the proof of Theorem 5 in \citep{fukumizu2007kernel} that it suffices to establish equations (14) and (15) in the supplement of \citep{fukumizu2007kernel} in our time-series setting, in order to prove Lemma \ref{lem:rhoconsistenth}. Equation (15) in the supplement of \citep{fukumizu2007kernel} is purely a population version which does not depend on the samples, so will go through in our case, too. Hence, we only need to show the validity of Equation (14) in the supplement of \citep{fukumizu2007kernel} for our setting, in order to complete the proof. This, in turn, follows from Corollary 1 and Theorem 5 in \citep{masry2011estimation}
, which gives the following for all $i\ne j$:
\begin{enumerate}
    \item $\hat{\sigma}_i^2 := \frac{1}{n}\sum_{t=1}^n (X_{ti}-\overline{X}^{(i)})^2 \rightarrow \sigma_i^2 := \mathrm{Var}(X_{1i})$, where $\overline{X}^{(i)} := \frac{1}{n} \sum_{t=1}^n X_{ti}$.
    \item $\hat{\rho}_{n;i,j} = \rho_{i,j} + O_P(n^{-1/2})$.
\end{enumerate}
The proof of Lemma \ref{lem:rhoconsistenth} is now complete.
\end{proof}

Once again, since $H_{n;i,j|\bm k}$ is a consistent estimator of $H_{i,j|\bm k}$, consistency of the PC algorithm using $H_{n;i,j|\bm k}$ as a statistic for testing conditional dependence in the non-Gaussian regime now follows from Theorems \ref{thm:skelconsistency} and \ref{thm:cpdagconsistency}.

\section{Analogous Results for Strongly Mixing Processes}\label{strmix}

The class of $\rho$-mixing processes is contained in the class of so called strongly mixing processes, and it turns out that all our results can easily be adopted in this more general framework, too. Let us first introduce the concept of strong mixing. For fixed $i,j \in 1,\ldots, p$, let $\mathcal{F}_{a}^{b}$ be the $\sigma$-field of events generated by the random variables $\{X_{ti},X_{tj}:a\leq t \leq b\}$, and $L_2(\mathcal{F}_a^b)$ be the collection of all second-order random variables which are $\mathcal{F}_a^b$-measurable. The stationary process $\{X_{ti},X_{tj}:t= 1,2,\ldots\}$ is called strongly mixing \citep{kolmogorov1960strong,bradley2005basic} if:

\[
\alpha_{ij}(k):=\sup_{l\geq 1}\sup_{\substack{A\in L_2(\mathcal{F}_{1}^l)\\ B\in L_2(\mathcal{F}_{l+k}^{\infty})}} |P(A\cap B) - P(A) P(B)| \rightarrow 0 \text{ as } k\rightarrow \infty.
\]
In this case, $\alpha_{ij}(k)$ are called the strongly mixing coefficients. 

All our results will go through for a strongly mixing process too, under the following slightly different set of assumptions:

\begin{enumerate}
 \item[(A.1)*] $\{X_{ti},X_{tj}: t=1,2,\ldots\}$ is strongly mixing for all $i,j$, with coefficients $\alpha_{ij}(k), k\ge 1$.
    \item[(A.2)*] $\e |X_{ti}|^{2\delta} < \infty$ for some $\delta > 2$ and all $t,i$, and the strongly mixing coefficients satisfy: $\sum_{k=1}^{\infty} \alpha_{ij} (k)^{1-2/\delta} < \infty$ for $i,j\in 1,\dots,p$.
    \item[(A.3)*] There exists a sequence of positive integers $s_n\rightarrow \infty$ and $s_n = o(n^{1/2})$ such that $n^{1/2}\alpha_{ij}(s_n)\rightarrow 0$ as $n\rightarrow\infty$ for $i,j\in 1,\ldots,p$.
    \item[(A.4)*] $P$ is faithful to a DAG $G$.
\end{enumerate}

The proof follows directly from \citep{masry2011estimation}, in a manner exactly similar to the $\rho$-mixing case, so we skip it. 

\begin{remark}
This section shows that one can derive consistency of the PC algorithm for the more general class of strongly mixing processes, if one agrees to assume conditions (A.1)* - (A.4)*. The results under the $\rho$-mixing assumption are still relevant, because if a $\rho$-mixing time series has finite fourth moment but all higher moments are infinite, then (A.2)* will not hold. In that case, one must appeal to the assumptions (A.1) - (A.4) in order to conclude consistency of the PC algorithm.      
\end{remark}

\section{Common Time Series Models}
We will consider two classes of examples to demonstrate the consistency of the PC algorithm in two commonly used time series models, namely VARMA and Linear Processes. Throughout this section, we are going to assume that the time series $\bm X_t$ satisfies $\e|X_{ti}|^{2\delta}<\infty$ for some $\delta >2$.
\subsection{VARMA Process} 
Assume now that $\bm{X}_t$ is a stationary vector autoregressive moving average (VARMA) process with values in $\mathbb{R}^p$. Then it admits a Markovian representation (See \citep{pham1985some}),

\begin{equation}\label{varmaeq}
    \bm{X}_t = HZ_t, ~~ Z_t = FZ_{t-1} + G \bm \epsilon_t
\end{equation}
where $Z_t$ are random vectors, $H, F, G$ are appropriate matrices with all eigenvalues of $F$ being strictly less than $1$, and $\bm \epsilon_t$ are i.i.d. error random vectors with density $g$. 

\begin{theorem}\label{thm:armacpdagconsistent}

Let $\bm{X}_t$ be a stationary VARMA process which is faithful with respect to a DAG $G$. Also, suppose that the density $g$ satisfies $\int \|x\|^{\beta_1} g(x)~dx <\infty$ and $\int \vert g(x) - g(x-\theta) \vert dx = O(\Vert \theta\Vert^{\beta_2})$ for some $\beta_1,\beta_2>0$. 
Then, we have:

\[
P(\hat{G}_{CPDAG,n}=G_{CPDAG})\rightarrow 1\quad \text{as}~ n\rightarrow \infty
\]
In particular, the conclusion holds for a stationary Gaussian VARMA process. 
\end{theorem}

\begin{proof} Note that Lemma \ref{lemma:armarhomixing} implies (A.1)* - (A.3)*. 
Hence, the consistency of the PC algorithm follows from the arguments presented in Section \ref{strmix}.  Theorem \ref{thm:armacpdagconsistent} now follows from Lemma \ref{lemma:armarhomixing}, on observing that the Gaussian density satisfies its hypothesis with $\beta_1=\beta_2 =1$.
\end{proof}

\subsection{Linear Process}

Assume now that $\bm{X}_t$ is a stationary linear process with values in $\mathbb{R}^p$. That is, there exist i.i.d. random vectors $\bm \epsilon_t$ having density $g$ (see \citep{pham1985some}), such that

\[
\bm{X}_t = \sum_{i=0}^{\infty}A_i\bm \epsilon_{t-i}, ~~ A_0 = I,
\]\label{stln}for matrices $A_i$, where $I$ is the identity matrix. We will assume that the density function $g$ of the errors $\bm \epsilon_t$ satisfies the condition $\int \vert g(x) - g(x-\theta)\vert ~ dx = O(\Vert \theta \Vert)$, which is proved in (\ref{eq:gaussnormtheta}) for a Gaussian density. 
Let the generating function of $A_i$ be denoted as $A(z)=\sum_{k=0}^{\infty}A_k z^k$. Let us also assume that $E\Vert \epsilon_i \Vert <\infty$. Then, the conditions of Theorem 2.1 in \citep{pham1985some} are satisfied. Denote $S_i=\sum_{j=i}^{\infty}\Vert A_j\Vert$, $\beta_{\lambda}(k)=\sum_{i=k}^{\infty}(S_i)^{\lambda/(1+\lambda)}$.

\begin{theorem}\label{thm:linearcpdagconsistent}
If $\bm{X}_t$ is a stationary linear process, which is faithful with respect to a DAG $G$, and  satisfies the following conditions:
\begin{enumerate}
    \item The density $g$ satisfies $\int \vert g(x) - g(x-\theta)\vert ~ dx = O(\Vert \theta \Vert)$,
    \item $A(z)\neq 0$ for $\vert z\vert\leq 1$,
    \item There exists $\lambda > 0$ such that $\sup_t \e\|\bm \epsilon_t\|^\lambda < \infty$,
    \item $\sum_{k=0}^{\infty} \beta_\lambda (k)^{1-2/\delta}<\infty$,
    \item $n^{1/2}\beta_\lambda(s_n)\rightarrow 0$ for some $s_n=o(n^{1/2})$.
\end{enumerate}

Then we have,

\[
P(\hat{G}_{CPDAG,n}=G_{CPDAG})\rightarrow 1\quad \text{as}~ n\rightarrow \infty
\]
\end{theorem}

For proving Theorem \ref{thm:linearcpdagconsistent}, we need the following lemma:

\begin{lemma}\label{lem:linearprocessmixing}
Under the assumptions of Theorem \ref{thm:linearcpdagconsistent},  
$\{X_{ti},X_{tj},t\geq 1\}$ is strongly mixing with coefficient $\alpha_{ij}(k)\leq C\beta_\lambda(k)$, where $C$ is a constant. Also, $\bm X_t,t\geq 1$ satisfies (A.2)* and (A.3)*.
\end{lemma}
\begin{proof}
By Theorem 2.1 in \citep{pham1985some}, the process is strongly mixing. Note that $\alpha_{ij}(k)\leq 4 \|\Delta_k\|_1$  and hence, $\alpha_{ij}(k)\leq C\beta_{\lambda}(k)$ (by Theorem 2.1 in \citep{pham1985some}). This proves Lemma \ref{lem:linearprocessmixing}. 
\end{proof}

\begin{proof}[Proof of Theorem \ref{thm:linearcpdagconsistent}]\label{rem:linearrhomixing} Theorem \ref{thm:linearcpdagconsistent} follows from its assumptions, together with Lemma \ref{lem:linearprocessmixing} and Section \ref{strmix}.
\end{proof}

\section{The Time-Aware PC algorithm}
The Time-Aware PC algorithm is an adaptation of the PC algorithm for causal inference in a time series scenario by considering DAGs between variables across different times instead of at a fixed time. This is especially relevant in practice where past observations in a time series have a causal influence on future observations, that is, there exist across-time causal relationships. Let us denote a time series of variables $V=\{1,\ldots,p\}$ by $\bm{X}_t=(X_{t1},\ldots,X_{tp}), t\geq 1$. An object of general interest for causal inference in the time series scenario is a graph $\bm{G}_R$ with nodes as the variables $V=\{1,\ldots,p\}$ and edge $u\rightarrow v$ if $X_{t_1 u}\rightarrow X_{t_2 v}$ for some $t_1\leq t_2$, where $X_{t_1 u}\rightarrow X_{t_2 v}$ are defined by a causal model with respect to a graph $\bm{G}$ with nodes $X_{t u}$. Some common examples are in interventional, structural and Granger Causality \citep{eichler2013causal,moraffah2021causal}, and applications in neurosciences \citep{biswas2022statistical1, smith2011network}, and econometrics \citep{granger2001essays}. We call $\bm{G}_R$ as the \emph{Rolled Graph} of $\bm{G}$.

Let $\bm{X}_1,\ldots,\bm{X}_n\in \mathbb{R}^p$ be a strictly stationary time series and $\bm{X}_t=(X_{t1},\ldots,X_{tp})$. Let the graph $\bm{G}=(\bm{V},\bm{E})$ consist of the set of nodes $\bm{V} = \left\{(v,t): v \in \{1,\ldots,p\}, t\in \{1,\ldots, \tau\}\right\}$ for $\tau\geq 1$, and directed edges $\bm{E}\subseteq \mathbf{V}\times \mathbf{V}$. Assume that $\bm{G}$ is a DAG. In $\bm{G}$, an edge $(u,t_1)\rightarrow (v,t_2)$ represents a connection from variable $u$ at time $t_1$ to variable $v$ at time $t_2$.

Let $\bm{\chi}_t = (\bm{X}_{1+(t-1)r},\ldots,\bm{X}_{\tau+(t-1)r})\in \mathbb{R}^{p\tau}, r\geq 1, t=1,2,\ldots,N:=\lfloor \frac{n-\tau}{r} \rfloor +1$. It follows that $\{\bm{\chi}_t\}_{t=1}^{\infty}$ is also strictly stationary. For our setting, the components of $\bm{\chi}_t$ correspond to the nodes in $\bm{G}$: the $(p(t'-1)+v)$-th component of $\bm{\chi}_{t}$ correspond to node $(v,t')$ in $\bm{G}$. 

In the traditional PC algorithm, (A.5) is assumed, that is, $\bm{X}_t$ is faithful with respect to a DAG $G$, which models causal relations at a fixed $t$ by the DAG $G$. In practice, this assumption can be unreasonable in a time series scenario which typically includes relationships between variables over time such as, $X_{1u}\rightarrow X_{2v}$. To better accommodate inter-temporal causal relations in a time-series scenario, we assume faithfulness with respect to $\bm{G}$ which has edges from variable $u$ at time $t_1$ to variable $v$ at time $t_2$. That is, instead of (A.4), we assume (B.1) as follows.

\begin{center}
(B.1): $\bm{\chi}_t$ is faithful with respect to such a DAG $\bm{G}$ as above. 
\end{center}

Moreover, if one is interested in using partial correlations as the tests for dependence, then the following condition on the time-series may be useful:

\begin{center}
(B.0): $\{\bm{X}_t\}_{t=1}^n$ is a strictly stationary Gaussian process. 
\end{center}

Under (B.1), let $\bm{G}_{CPDAG}=(\bm{V},\bm{E}(\bm{G}_{CPDAG}))$ be the $CPDAG$ from the DAG $\bm{G}$. Next, $\bm{G}_{CPDAG}$ is transformed to obtain $\bm{G}_{CPDAG,R}$ with nodes $1,\ldots,p$, and directed edges $\bm{E}(\bm{G}_{CPDAG,R})$ such that $u\rightarrow v \in \bm{E}(\bm{G}_{CPDAG,R})$ if and only if $(u,t_1)\rightarrow (v,t_2) \in \bm{E}(\bm{G}_{CPDAG})$ for some $1\leq t_1 \leq t_2 \leq \tau$. $\bm{G}_{CPDAG,R}$ is the Rolled Graph of $\bm{G}_{CPDAG}$ and referred to as the Rolled Markov Graph with respect to $\bm{\chi}_t$.

The Time-Aware PC Algorithm uses PC Algorithm to estimate $\bm{G}_{CPDAG}$ based on $\bm{\chi}_1,\ldots,\bm{\chi}_n$ by $\hat{\bm{G}}_{CPDAG}$. Next, the estimate $\hat{\bm{G}}_{CPDAG}$ is transformed to obtain $\hat{\bm{G}}_{CPDAG,R}$. See \citep{biswas2022statistical2} for more details.

\begin{algorithm}
\SetKwInOut{Input}{Input}
\SetKwInOut{Output}{Output}
\Input{Vertex Set $\bm{V}$, $\{\bm{\chi}_t\}_{t=1}^{N}$}
\Output{$\hat{\bm{G}}_{CPDAG,R}$ (estimate of $\bm{G}_{CPDAG,R}$)}

Use PC Algorithm to estimate $\bm{G}_{CPDAG}$ based on $\bm{\chi}_1,\ldots,\bm{\chi}_N$ by $\hat{\bm{G}}_{CPDAG}$. 

Next, $\hat{\bm{G}}_{CPDAG}$ is transformed to obtain $\hat{\bm{G}}_{CPDAG,R}$.
\caption{Time-Aware PC Algorithm}
\end{algorithm}

For any DAG $\tilde{\bm G}$ with nodes $\bm{V}$, recall that $pa_{R,\tilde{\bm G}}(v)$ denotes the parents of $v$ in the Rolled Graph of $\tilde{\bm G}$. We show that the parents of a node $v$ in a Rolled Markov Graph with respect to $\bm{\chi}_t$ is the union of parents of $v$ in the Rolled Graphs from a Markov equivalence class. For any graph $\tilde{\bm{G}}=(\bm{V},\bm{E}(\tilde{\bm{G}}))$ in the Markov equivalence class $\mathcal{M}(\bm{G})$. Finally, let $pa_{R}(v)$ denote the parent set of $v$ in $\bm{G}_{CPDAG,R}$. Below, we show that $pa_{R}(v)$ can be expressed as the union of $pa_{R,\tilde{\bm G}}(v)$ over all rolled graphs $\tilde{\bm G}$ in the Markov equivalence class of $\bm G$.
\begin{lemma}
If $\bm{\chi}_t$ is faithful with respect to $\bm{G}$, then
\[
pa_{R}(v) = \bigcup_{\tilde{\bm{G}}\in\mathcal{M}(\bm{G})} pa_{R,\tilde{\bm G}}(v)
\]
\end{lemma}
\begin{proof}
Choose $v'\in pa_{R,\tilde{\bm G}}(v)$ for some $\tilde{\bm{G}}\in\mathcal{M}(\bm{G})$. Then, by definition, $(v',t')\rightarrow (v,t)\in \bm{E}(\tilde{\bm{G}})$ for some $t'\in t-\tau+1,\ldots, t$. Therefore, $(v',t')\rightarrow (v,t)\in \bm{E}(\bm{G}_{CPDAG})$. So, $v'\rightarrow v \in \bm{G}_{CPDAG,R}$. Therefore, $v'\in pa_{R}(v)$. This shows that $\bigcup_{\tilde{\bm{G}}\in\mathcal{M}(\bm{G})} pa_{R,\tilde{\bm G}}(v)\subseteq pa_{R}(v)$.

For the other direction, choose $v'\in pa_{R}(v)$, whence $v'\rightarrow v\in \bm{E}(\bm{G}_{CPDAG,R})$. Therefore, there exists $t'\leq t$ such that $(v',t')\rightarrow (v,t)\in \bm{E}(\bm{G}_{CPDAG})$. Hence there exists $\tilde{\bm{G}}\in \mathcal{M}(\bm{G})$ such that $(v',t')\rightarrow (v,t)\in \bm{E}(\tilde{\bm{G}})$. By definition $v'\in pa_{R,\tilde{\bm G}}(v)$. Therefore, $pa_{R}(v) \subseteq \bigcup_{\tilde{\bm{G}}\in\mathcal{M}(\bm{G})} pa_{R,\tilde{\bm G}}(v)$.
\end{proof}

\subsection{Consistent Estimation with Time-Aware PC}
In this section, we show consistency of the Time-Aware PC algorithm.

\begin{theorem}[Consistency of Time-Aware PC algorithm]\label{ctpct}
Assume (B.1), (A.1)-(A.3) (or (A.1)*-(A.3)*). Denote by $\hat{\bm{G}}_{CPDAG,n}$ the estimate of $\bm G_{CPDAG}$ by the PC algorithm based on $\bm{\chi}_{1},\ldots,\bm{\chi}_{n-\tau+1}$. Then,
\[
P(\hat{\bm{G}}_{CPDAG,n}=\bm{G}_{CPDAG})\rightarrow 1 \text{ as }n\rightarrow \infty
\]
\end{theorem}
\begin{proof}
By Lemma \ref{lemma:chirhomixing}, if $\bm{X}_t$ satisfies (A.1)-(A.3) (or (A.1)*-(A.3)*), then so does $\bm{\chi}_t$. Also, note that (B.1) is same as (A.4) with respect to $\bm{\chi}_t$ being faithful with $\bm{G}$. The statement then follows from Theorem \ref{thm:cpdagconsistency}.
\end{proof}

\begin{corollary} Since $\hat{\bm{G}}_{CPDAG,n}=\bm{G}_{CPDAG}\implies \hat{\bm{G}}_{CPDAG,R,n}=\bm{G}_{CPDAG,R}$, therefore, also, $P(\hat{\bm{G}}_{CPDAG,R,n}=\bm{G}_{CPDAG,R})\rightarrow 1$ as $n\rightarrow \infty$.  
\end{corollary}
\begin{remark}[Subsampled Time-Aware PC for non-stationary processes (TPCNS)] In practice, the entire time series may not be strictly stationary, in such a scenario, the Time-Aware PC algorithm is conducted in a shorter time-window of length $L$ selected at random, i.e. $[n_i, n_i+L]\subseteq \{1,\ldots,\lfloor \frac{n-\tau}{r} \rfloor+1\}$, and a set of graphs $\hat{\bm{G}}_{CPDAG}^{(i)}$ is estimated with $\bm{\chi}_t, t\in [n_i, n_i+L]$. If there is an edge $(u,t_1)\rightarrow (v,t_2)$ for $t_1>t_2$ in $\hat{\bm{G}}_{CPDAG}^{(i)}$ then the edge is re-oriented to $(u,t_1)\rightarrow (v,t_2)$ in $\hat{\bm{G}}_{CPDAG}^{(i)}$, which is then transformed to $\hat{\bm{G}}_{CPDAG,R}^{(i)}$. Next, a single graph $ \hat{\bm{G}}_{CPDAG,R}^*$ is obtained by union of $\hat{\bm{G}}_{CPDAG,R}^{(i)}$ while excluding those edges with exceedingly low frequency of occurrence. At the last step, $\hat{\bm{G}}_{CPDAG,R}^{(i)}$ is pruned by excluding those edges with exceedingly low edge weights which can be obtained by interventional causal effects.
\end{remark}
\begin{example}[VARMA process]
Assume now that $\bm{X}_t$ is the stationary VARMA process \eqref{varmaeq} satisfying the assumptions of Theorem \ref{thm:armacpdagconsistent}. In this example, consistency of the Time-Aware PC algorithm holds under B.1, by Theorem \ref{ctpct}.


\end{example}
\begin{example}[Linear Process] 
Similarly, if $\bm X_t$ is the stationary linear process \eqref{stln} satisfying the assumptions of Theorem \ref{thm:linearcpdagconsistent}, then once again, consistency of the Time-Aware PC algorithm holds under B.1, by Theorem \ref{ctpct}.
\end{example}

\section{Simulation Studies}
We compare the performance of Time-Aware PC for stationary data (TPCS), subsampled Time-Aware PC for non-stationary data (TPCNS), the usual PC algorithm using partial correlation test as well as Hilbert-Schmidt test for conditional independence (denoted TPCSHS, TPCNSHS and PCHS for the latter), and Granger Causality (GC) to recover the ground truth causal relations from four simulation paradigms. The simulation paradigms correspond to specific model assumptions to assess the impact of model specifications on the performance of the approaches (See Appendix \ref{simul:details}).

\begin{figure}
    \centering
    \includegraphics[width=\textwidth]{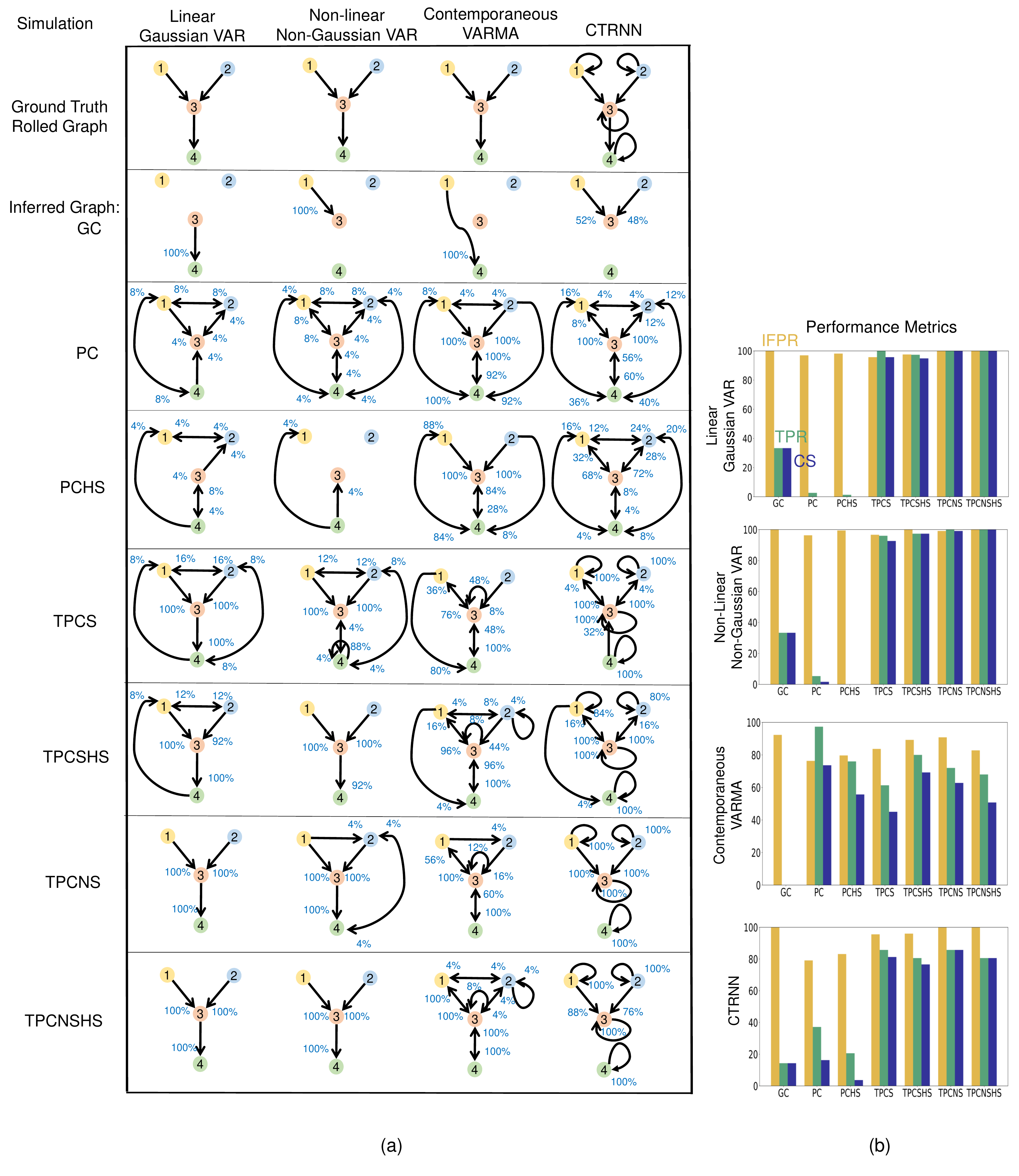}
    \caption{(a) The outcomes of GC, PC, PCHS, TPCS, TPCHS, TPCNS and TPCNSHS are compared on four examples of motifs and simulation paradigms; from left to right: Linear Gaussian VAR, Non-linear Non-Gaussian VAR, Contemporaneous VARMA and CTRNN. Table: 4-neurons motifs that define the Ground Truth (row 1) are depicted along  with inferred graphs over simulation instances by the different methods (row 2-8). Each inferred graph has an edge $v\rightarrow w$ that corresponds to an edge detected in any of the inference instances. The percentage (blue) next to each edge indicates the number of times the edge was detected out of all instances. (b) $\text{IFPR}$ (gold) , TP rate (green) and Combined Score (blue) of each method are shown for each motif.}
    \label{fig:simul_eval}
\end{figure}

We generated 25 simulated time series from each simulation paradigm and estimated the causal relationships from each time series. We summarized the performance of the methods in recovering the ground truth causal relationships using the following three metrics: (1) Combined Score (CS), (2) True Positive Rate (TPR), (3) 1 - False Positive Rate (IFPR). Let True Positive (TP) represent the number of correctly identified edges, True Negative (TN) represent the number of correctly identified missing edges, False Positive (FP) represent the number of incorrectly identified edges, and False Negative (FN) represent the number of incorrectly identified missing edges across simulations. IFPR is defined as: $$\text{IFPR}=\left(1-\frac{\text{FP}}{\text{FP+TN}}\right)\cdot 100,$$ which measures the ratio of the number of correctly identified missing edges by the algorithm to the total number of true missing edges. Note that the rate is reported such that $100\%$ corresponds to no falsely detected edges. TPR is defined as $\text{TPR}=\left(\frac{\text{TP}}{\text{TP} + \text{FP}}\right)\cdot 100$ as i.e. the ratio of the number of correctly identified edges by the algorithm to the total number of true edges in percent. The Combined Score (CS) is given by Youden's Index \citep{vsimundic2009measures,hilden1996regret}, as follows, $\text{CS} = \text{TPR} - \text{FPR}$.

In implementations of the PC algorithm in the \emph{pcalg} package in \emph{R} and other software such as \emph{TETRAD$^{IV}$}, the Gaussian conditional dependence tests use a fixed level $\alpha$ and $\sqrt{n-\bm{k}-3} ~ g(\hat{\rho}) \leq \Phi^{-1}(1-\alpha)$, thereby we use $\gamma = \Phi^{-1}(1-\alpha)/\sqrt{n-\bm{k}-3}$ which also gives rise to a consistent test. The PC algorithm with Hilbert-Schmidt conditional independence criterion is implemented using the \emph{kpcalg} package in \emph{R}.

The true graphs in each of the simulations consist of $4$ nodes and $16$ possible edges (including self-loops) between the nodes. Therefore for the $25$ simulations in a simulation setting, there are a total of $25$ graphs to infer and $400$ possible edges. Figure \ref{fig:simul_eval} shows the ground truth graph in each simulation setting for noise standard deviation $\eta =1$ and performance of the estimation for $\alpha = 0.05$. We report the percentage of the simulations in which an edge is estimated to be present and also compare the approaches in estimating the true edges by TPR, IFPR and CS scores of performance. A higher percentage indicates higher confidence in the detection of that edge. Figure \ref{fig:comp_noise} compares the Combined Score of performance of the approaches over different values of $\alpha$ and noise standard deviation $\eta$ in each simulation setting.

\begin{itemize}[leftmargin=0pt]
    \item[] In \textit{Linear Gaussian VAR (left column in Figure \ref{fig:simul_eval})}, GC estimates a single edge $3\rightarrow 4$ among the three edges of the true graph in every iteration leading to IFPR $100\%$, TPR $33.3\%$ and CS $33.3\%$. In comparison, PC and PCHS do not estimate any edge more than $10\%$ of the iterations. Overall, PC has IFPR $96.9\%$, TPR $2.7\%$ and CS $-0.4\%$ and closely followed by PCHS. TPCS and TPCHS estimates the true edges for $1\rightarrow 3$ and $3\rightarrow 4$ for $100\%$ of the iterations, and the true edge $2\rightarrow 3$ in $100\%, 92\%$ of the iterations for TPCS and TPCHS respectively. There are other spurious edges also obtained in $8\%$-$16\%$ of the iterations. Overall, TPCS has IFPR $95.7\%$, TPR $100\%$, CS $95.7\%$ and TPCHS has IFPR $97.5\%$, TPR $97.3\%$ and CS $94.9\%$. TPCNS and TPCNSHS obtains the Ground Truth, with no spurious edges and obtains the true edges in all of the trials ($1\rightarrow 3, 2\rightarrow 3, 3\rightarrow 4$ with $100 \%$, $100 \%$ and $100 \%$ respectively), overall both having IFPR, TPR and CS of $100\%$. Thereby, among the three methods, we conclude that TPCNS, TPCNSHS, TPCS and TPCSHS have improvement in performance by a great margin ($+66.7\%$) in CS while TPCNS and TPCNSHS detects the true edges perfectly. The non-detection of the edges by PC algorithm can be explained by Theorem \ref{thm:armacpdagconsistent}, since at a fixed time the simulated process is faithful with respect to the empty graph, even though there are across-time causal relations leading to the ground truth rolled graph as in Figure \ref{fig:simul_eval}. 
    \item[] In the \textit{Non-linear Non-Gaussian VAR (second column)}, as previously, GC always detects a single edge $1\rightarrow 3$ with $100\%$ out of the three true edges. PC and PCHS does not detect any edge with greater than $8\%$ occurrence in iterations. In contrast, TPCS outputs the true edges with a high percentage ($1\rightarrow 3$, $2\rightarrow 3$,$3\rightarrow 4$ with $100\%, 100\%, 88\%$) and some spurious edges with a low percentage of less than $12\%$, while TPCSHS outputs the true edges with a high percentage ($1\rightarrow 3$, $2\rightarrow 3$,$3\rightarrow 4$ with $100\%, 100\%, 92\%$) and without spurious edges. TPCNS detects the true edges in all the trials ($1\rightarrow 3$, $2\rightarrow 3$,$3\rightarrow 4$ with $100\%, 100\%, 100\%$) and relatively lower spurious edges in $4\%$ of the trials. TPCNSHS detects the true edges in all the trials ($1\rightarrow 3$, $2\rightarrow 3$,$3\rightarrow 4$ with $100\%, 100\%, 100\%$) and no spurious edges. In summary, all approaches yielded $\text{IFPR}$ close to $100\%$, while for $\text{TPR}$ TPCNSHS and TPCNS outperformed the other approaches with TPR $100\%$, followed by TPCSHS and TPCS with $97\%,96\%$ respectively, and GC with $33.3\%$, and for CS also, TPCNSHS and TPCNS outperformed at $100\%, 99.1\%$, closely followed by TPCSHS and TPCS with $97.3\%, 92.6\%$ and GC at $33.3\%$. For this scenario, TPCNSHS has the best performance followed closely by other variants of TPC such as TPCSNS, TPCSHS and TPCS among the methods.
    \item[] In \textit{Contemporaneous VARMA (third column)}, the ground truth causal graph encode the linear dependence between variables at each fixed time as well as weaker relationship following the same rolled graph across time. GC obtains a spurious edge $1\rightarrow 4$ for $100\%$ of the iterations. PC estimates the true edges $1\rightarrow 3, 2\rightarrow 3, 3\rightarrow 4$ in $100\%,100\%,92\%$ of the trials, TPCS in $76\%,8\%,100\%$, TPCSHS in $96\%, 44\%, 100\%$, and TPCNS in $100\%,16\%,100\%$ trials with other spurious edges also detected. In summary, GC, PC, TPCS and TPCNS yielded $\text{IFPR} = 92.3\%, 76.3\%, 83.7\%, 90.7\%$, $\text{TPR} = 0\%, 97.3\%, 61.3\%, 72.0\%$ and CS $=-7.7\%, 73.5\%,45.0\%, 62.8\%$. TPCHS performed better than TPCS with $\text{IFPR}$ $76\%$, $\text{TPR}$ $80\%$, $\text{CS}$ $80\%$. For this scenario, PC has the highest performance, closely followed by TPCSHS, PCHS, TPCNS, TPCNSHS, TPCS and then followed by GC. This is an impractical scenario when the DAG encoding causal relations at each fixed time is also the same as the Ground Truth Rolled Graph, and the PC algorithm is shown to be consistent in Theorem \ref{thm:armacpdagconsistent} for the causal relations at each fixed time.
    \item[] In the \textit{CTRNN scenario (fourth column)}, self-loops are present for each neuron. GC obtains two of the three true non-self edges $1\rightarrow 3, 2\rightarrow 3$  for $52\%, 48\%$ of the trials. PC detects spurious edges in up to $12\%$ of the trials, but also infers the non-self true edges $1\rightarrow 3, 2\rightarrow 3$ for $100\%$ of the trials. TPCS outputs all the self true edges and non-self true edges $1\rightarrow 3, 2\rightarrow 3$ for $100\%$ of the trials while also detecting false edges in relatively few $4\%$ of the trials. In comparison, TPCNS infers no false edges and all the self true edges for $100\%$ of the trials and non-self true edges $1\rightarrow 3$ and $2\rightarrow 3$ for $100\%, 100\%$ of the trials. TPCNSHS also infers no false edges and all the self true edges for $100\%$ of the trials and non-self true edges $1\rightarrow 3$ and $2\rightarrow 3$ for $88\%, 78\%$ of the trials. In summary, IFPR of GC, PC, TPCS and TPCNS are $100\%, 79.1\%, 95.6\%, 100\%$, that for $\text{TPR}$ are $14.3\%, 37.1\%, 85.7\%, 85.7\%$, and that for CS are $14.3\%, 16.3\%, 81.3\%, 85.7\%$ respectively and the HS versions closely follow in performance. Among all methods, TPCS and TPCNS have the highest TPR together, followed by PC and lastly GC. TPCNS and GC have the highest IFPR having not detected any false edges, followed by TPCS and then PC. In terms of CS, TPCNS and TPCNSHS have the highest performance, followed by TPCS and TPCSHS, compared to other methods.
\end{itemize}

\begin{figure}
    \centering
    \includegraphics[width=0.9\textwidth]{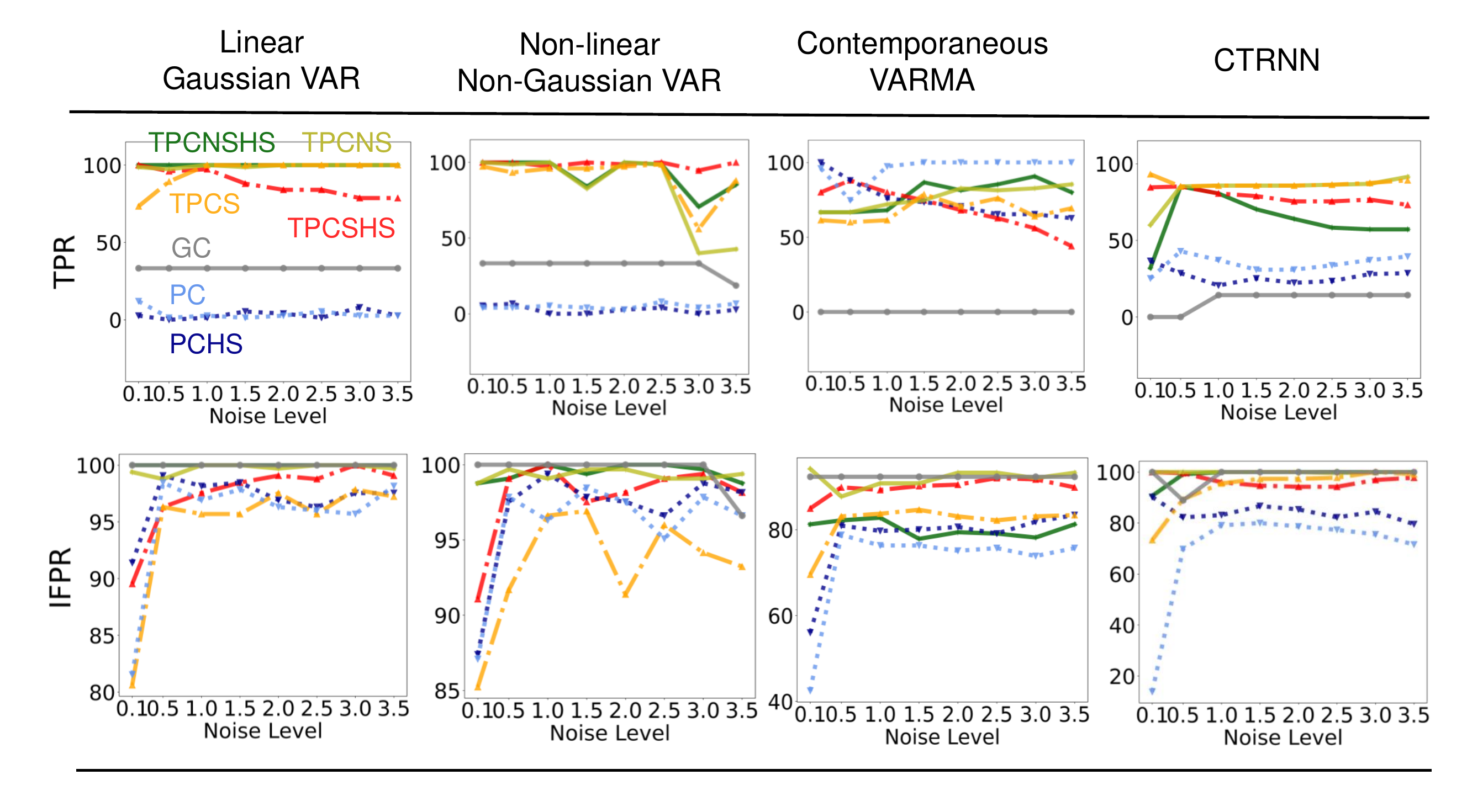}
    \caption{Performance of the algorithms - TPCNSHS (dark green), TPCNS (light green), TPCSHS (red), TPCS (orange), PCHS (dark blue), PC (light blue), GC(gray), as measured by TPR and IFPR in different simulation paradigms: Linear Gaussian VAR, Non-linear Non-Gaussian VAR, Contemporaneous VARMA and CTRNN, over increasing levels of noise variance and $\alpha=0.05$.}
    \label{fig:tprfpr}
\end{figure}
\begin{figure}
    \centering
    \includegraphics[width=\textwidth]{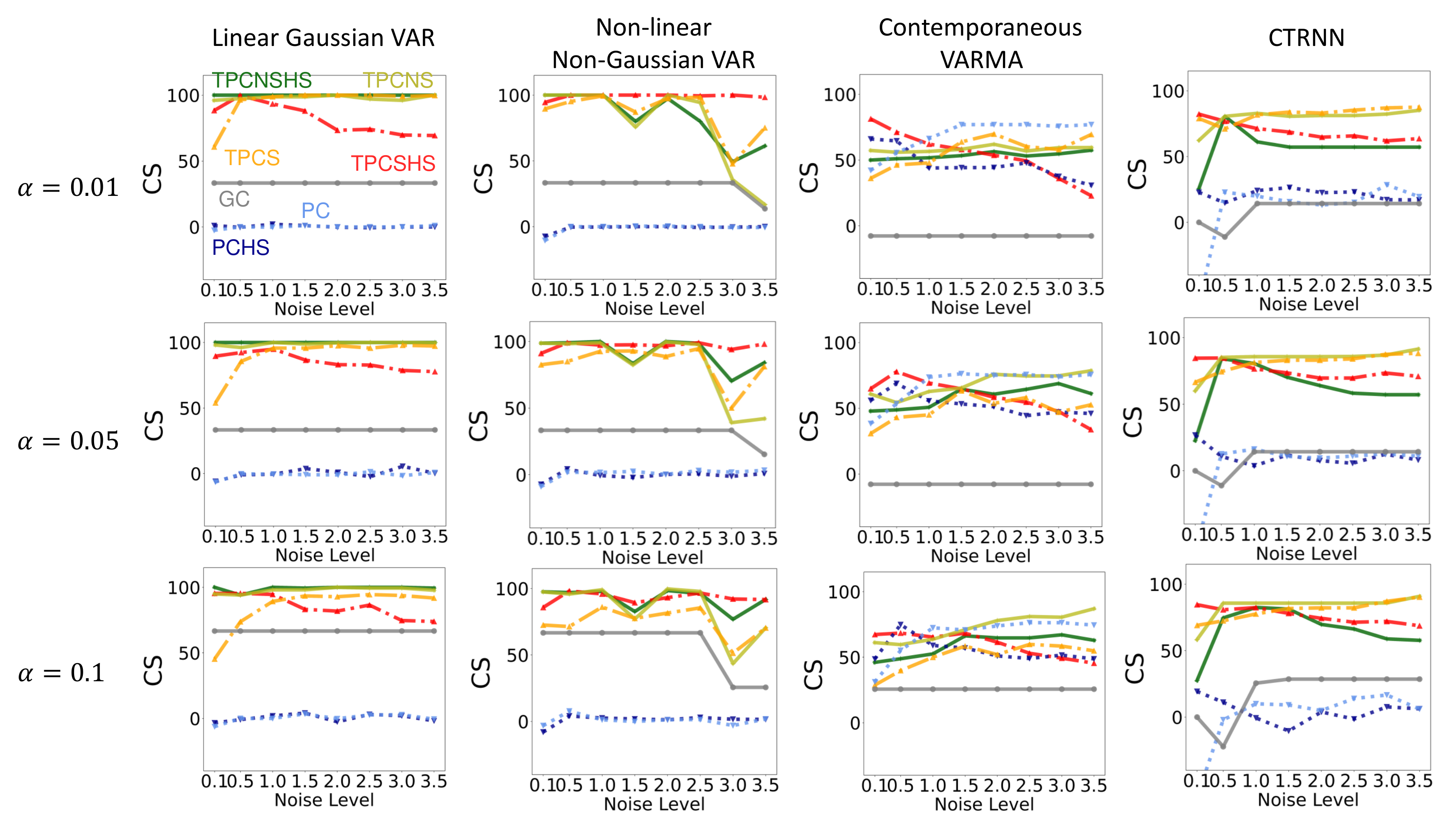}
    \caption{Combined Score of the four methods - TPCNSHS (dark green), TPCNS (light green), TPCSHS (red), TPCS (orange), PCHS (dark blue), PC (light blue), GC(gray), over varying noise levels in simulation $\eta = 0.1,0.5,1.0,\ldots,3.5$, for simulated motifs from Linear Gaussian VAR, Non-linear Non-Gaussian VAR, Contemporaneous VARMA and CTRNN paradigms (left to right), with $\alpha = 0.01, 0.05, 0.1$ for conditional dependence tests (top to bottom).}
    \label{fig:comp_noise}
\end{figure}

We compare the TPR and IFPR of the approaches across varying levels of noise standard deviation $\eta$ from $0.1$ to $3.5$ and $\alpha = 0.05$ in Figure \ref{fig:tprfpr}. In the Linear Gaussian scenario, we note that TPCNS and TPCNSHS has a TPR of $\approx 100\%$ across all levels of simulation noise, and is followed by TPCS and TPCHS, and then by GC in performance and lastly PC. In terms of IFPR also, TPCNSHS and TPCNS have the best performance overall. In the Non-linear Non-Gaussian scenario, TPCSHS has a TPR of $\approx 100\%$, while TPCNSHS has some fluctuations below $100\%$ at a few noise levels, and all the TPC variants have a TPR exceeding that of PC and GC. In terms of IFPR, TPCNSHS and TPCNS have an IFPR of $\approx 100\%$ and exceeds the other approaches. In the Contemporaneous scenario, PC and TPCNS have comparatively higher TPR while TPCNSHS, TPCNS and TPCSHS have comparatively higher IFPR than other approaches. In the CTRNN scenario, both the TPR and IFPR of TPCS and TPCNS are higher compared to the other methods over the different noise levels. We also compare the Combined Score of the approaches across varying levels of simulation noise $\eta$ from $0.1$ to $3.5$ and thresholding parameter $\alpha = 0.01,0.05,0.1$ in Figure \ref{fig:comp_noise}. In the Linear Gaussian scenario, we note that TPCNSHS and TPCNS have a CS of $\approx 100\%$ across all levels of simulation noise and thresholding parameter $\alpha$, and is closely followed by TPCS, then TPCSHS, GC, and lastly PC. In the Non-linear Non-Gaussian scenario, TPCSHS has the highest CS, closely followed by TPCS, TPCNS, and TPCNSHS, compared to other methods across levels of noise and $\alpha$. In the Contemporaneous scenario, PC has better CS for $\alpha = 0.01$ and for $\alpha=0.05,0.1$ TPCNS exceeds in performance for higher levels of noise, closely followed by PC and other variants of TPC and then GC. In the CTRNN scenario, TPCNS and TPCS have the highest CS overall followed by TPCNSHS and TPCSHS, compared to the other methods.

\section{Application to River Runoff Benchmark Data}
We also conducted analysis of a real benchmarking dataset in the public benchmarking platform - \emph{CauseMe} \citep{bussmann2021neural}. We used the \emph{River Runoff} dataset (See Appendix \ref{data:bm}). The River Runoff dataset is expected to have contemporaneous as well as across-time interactions with noise and  interaction strengths not controlled for, thereby would demonstrate the performance of the methods in an empirical setting. We compare the approaches selVAR \citep{pmlr-v123-weichwald20a}, SLARAC \citep{weichwald2020causal}, PCMCI-GPDC \citep{runge2019detecting} - which are among the top of the leaderboard for performance on the benchmarking dataset, GC, PC, TPC (Our) and TPCNS (Our). We used 1 - False Positive Rate (IFPR), True Positive Rate (TPR) and Combined Score given by Youden's Index (CS = TPR - FPR) (See Table \ref{tab:causalsummarybm}) to measure the performance of the algorithms.

\begin{table}[h]
\centering
\begin{tabular}{lll}\toprule
\multicolumn{3}{c}{\textbf{Combined [True, 1-False] Rates (\%)}}\\\midrule
     \textbf{Algorithm} & \textbf{River-Runoff (Real)}\\\midrule
    GC  & ~~~~~37 [45, 92]\\
    PC & ~~~~~38 [55, 83]\\
    PCMCI & ~~~~~45 [100, 45]\\
    SLARAC & ~~~~~50 [86, 64]\\
    PCHS & ~~~~~50 [64, 86]\\
    TPCSHS & ~~~~~50 [64, 87]\\
    selVAR & ~~~~~54 [91, 63]\\
    TPCS & 68(+14\%) [82, 86]\\
    TPCNS &\textbf{72(+18\%)} [82, 90]\\
    TPCNSHS & \textbf{72(+18\%)} [100, 72]\\\bottomrule
\end{tabular}
\vspace{2mm}
\caption{Comparison of performance on benchmarking datasets. For each dataset, each method's Combined Score, True Positive Rate, and 1-False Positive Rate are reported (Higher value is better).}\label{tab:causalsummarybm}
\end{table}
In terms of performance with respect to CS, TPCNS yielded the best performance with a CS of $72\%$, closely followed by TPCS at $68\%$ and selVAR, SLARAC, PCMCI-GPDC, PC and GC at $54\%,50\%,45\%, 38\%, 37\%$ respectively. TPCNS and TPCS exceeds the best among the existing approaches \textbf{by 18\%} and $14\%$ respectively. With respect to TPR, both TPCNS and PCMCI achieve the highest TPR at $100\%$, followed by selVAR, SLARAC, TPCNS, TPCS, PC and GC at $91\%,86\%,82\%,82\%,55\%, 45\%$ respectively. With respect to IFPR, GC has the best performance with a score of $92\%$ closely followed by TPCNS and TPCS at $90\%,86\%$, PC at $83\%$, and SLARAC, selVAR and PCMCI with at $64\%,63\%, 45\%$ respectively. PC and GC have a higher IFPR than one of TPCS or TPCNS because they detect less number of false edges, but that is achieved at the expense of detecting a less number of edges altogether including a less number of true edges which in turn leads to their TPR among the lowest in Table \ref{tab:causalsummarybm}. In comparison, TPCNS and TPCS detects true edges with greater sensitivity, thereby achieving a higher TPR. TPCS and TPCNS better maintains the trade-off between TPR and FPR, thereby leading to a better CS in comparison to the other methods. 

Since the coupling between variables as well as noise are not controlled and contemporaneous interactions are expected to be present as the sampling resolution is greater than the time taken for interactions between the variables, the real dataset provides a challenge for the methods. TPCS and TPCNS outperform other approaches by a CS of $14\%$ and $18\%$ respectively and shows significant improvement in performance than the other methods. This performance in a benchmarked real data setting demonstrates the applicability of TPC to real data scenarios.
\section{Conclusion}
We show that the PC algorithm is consistent for estimating the equivalence class of the DAG between variables from a stationary Gaussian time series with $\rho$-mixing properties, thereby demonstrating the asymptotic correctness of the PC algorithm in a dependent sampling setting. This enables us to show consistency of the PC algorithm in popular time series such as vector auto-regressive models and linear processes.

We also show that the Time-Aware PC algorithm consistently estimates the equivalence class of the DAG whose nodes are variables at different times using the PC algorithm and then transforms it into a rolled Markov graph between the variables. In contrast to PC, the consistency of the Time-Aware PC algorithm assumes faithfulness with respect to a DAG whose nodes are tuples of variable and time. The latter indicates causal relations between variables at different times. Therefore, in the time series scenario, while the PC algorithm supports finding of contemporaneous causal relations between the variables, Time-Aware PC provides a more general setting in which it first finds causal relations between the variables across different times and then transforms the temporal causal relations into the rolled Markov graph between the variables.

We compare the performance in recovery of ground truth causal relations for PC,  GC and Time-Aware PC for stationary and non-stationary time series in simulation studies with different model assumptions and benchmark real data. The results indicate the greater performance of Time-Aware PC in general when interactions in the time series are across different times and greater performance of PC in the specific scenario when the interactions are restricted to be contemporaneous.

\section{Acknowledgement}
Somabha Mukherjee was supported by the National University of Singapore Start-Up Grant R-155-000-233-133, 2021. The authors thank Eli Shlizerman for his helpful advice.
\bibliographystyle{unsrt}
\bibliography{main}

\appendix

\section{Facts about the VARMA Process \eqref{varmaeq}
}
\begin{lemma}\label{lemma:armarhomixing}
Suppose that the eigenvalues of $F$ in \eqref{varmaeq} are of modulus strictly less than 1. Then $\{X_{ti},X_{tj},t\geq 1\}$ is strongly mixing at an exponential rate for all $i,j$, that is, $\xi_{ij}(t)=\exp(-a_{ij}t)$ for some $a_{ij}>0$. Furthermore, the conditions $\int \| x\| g(x) dx < \infty$ and $\int \vert g(x+\theta) - g(x) \vert dx = O(\|\theta\|)$ hold for the Gaussian density $g$.
\end{lemma}
\begin{proof}
It follows from Theorem 3.1 in \citep{pham1985some}, that $\Vert \Delta_t\Vert_1\rightarrow 0$ at exponential rate, where $\Vert \Delta_t \Vert_1$ is a Gastwirth and Rubin mixing coefficient of $\bm{X}_t,t\geq 1$. Furthermore, $\alpha_{ij}(t)\leq 4\Vert \Delta_t \Vert_1$ (See \citep{pham1985some}), where $\alpha_{ij}(t)$ are the strong mixing coefficients of $\bm{X}_t,t\geq 1$. Hence, $\alpha_{ij}(t)\rightarrow 0$ as $t\rightarrow \infty$, at an exponential rate.

Now, $\int \| x\| g(x) dx < \infty$ holds trivially for the Gaussian density $g$. Therefore, it only remains to show that $\int \vert g(x) - g(x-\theta) \vert dx = O(\Vert \theta\Vert)$ for the Gaussian density. Towards this, note that:

\begin{align*}
    \int \vert g(x+\theta) - g(x) \vert dx &= \int\left\vert\int_{0}^{1} \nabla g(x+z\theta) dz\cdot \theta\right\vert dx \\
    &\leq \Vert \theta \Vert\int \int_{0}^{1} \Vert \nabla g(x+z\theta)\Vert  dz ~dx \\
    &= \Vert \theta \Vert \int_{0}^{1} \int  \Vert \nabla g(x+z\theta)\Vert dx~ dz\\
    &= \Vert \theta \Vert \int_{0}^{1} \int  \Vert \nabla g(x)\Vert dx~ dz\\
    &\text{by change of variable from $x+z\theta$ to $x$}\\&\lesssim \Vert \theta \Vert\int_0^1 \sup_i \mathbb{E}_{\bm Z\sim N(\bm \mu,\Sigma)} |Z_i|~dz = O(\|\theta\|). \numberthis\label{eq:gaussnormtheta}
\end{align*}
where $a \lesssim b$ denotes that $a \le C b$ for some universal constant $C>0$.
\end{proof}
\begin{lemma}\label{lemma:chirhomixing}
If $\bm{X}_t$ satisfies (A.1)-(A.3) (or (A.1)* - (A.3)*), then so does $\bm{\chi}_t$.
\end{lemma}
\begin{proof}
We only address the case when $\bm{X}_t$ satisfies (A.1)-(A.3), as the case when $\bm{X}_t$ satisfies (A.1)*-(A.3)* can be proved exactly similarly. For $u\in \{i,j\}$, let $q_u$ and $r_u$ denote the quotient and remainder (respectively) on dividing $u$ by $p$, with the slightly different convention of redefining $r_u$ to be $p$ if $r_u=0$. Note that
$\chi_{tu} = X_{(t-1)r+q_u +\boldsymbol{1}(r_u\ne p)~,~r_u}.$ Hence, if we define $\mathcal{G}_a^b$ as the $\sigma$-field generated by the random variables $\{\chi_{si},\chi_{sj}:a\le s\le b\}$, then assuming $i\le j$ without loss of generality, we have:
$$\mathcal{G}_{1}^l \subseteq \mathcal{F}_{1}^{(l-1)r+q_j+1}(r_i,r_j)\quad\text{and}\quad \mathcal{G}_{l+k}^{\infty} \subseteq \mathcal{F}_{(l+k-1)r+q_i}^{\infty}(r_i,r_j),$$
where $\mathcal{F}_{a}^b(u,v)$ denotes the $\sigma$-field generated by the random variables $\{X_{tu},X_{tv}: a\le t\le b\}.$
Hence, denoting $\tilde{\xi}_{ij}(k)$ to be the maximal correlation coefficients for the process $\{\chi_{ti},\chi_{tj}: t=1,2,\ldots\}$, we have:
$$\tilde{\xi}_{ij}(k) \le \xi_{r_i,r_j}(kr-(q_j-q_i)-1)$$ for all $k > (q_j-q_i+1)/r$. Assumptions (A.1) and (A.2) for the process $\{\chi_t\}$ now follow immediately. For (A.3), note that if $n^{1/2}\xi_{r_i,r_j}(s_n) \rightarrow 0$ for some $s_n=o(n^{1/2})$, then $t_n := r^{-1}(s_n+(q_j-q_i)+1)$ satisfies $t_n = o(n^{1/2})$, and $n^{1/2}\tilde{\xi}_{ij}(t_n) \rightarrow 0$, thereby verifying assumption (A.3) for the process $\{\chi_t\}$. 
\end{proof}


\section{Simulation Study Details} \label{simul:details}
We study the following simulation paradigms.
\begin{enumerate}[wide=0pt]
\item \underline{Linear Gaussian Vector Auto-Regressive (VAR) Model (Figure \ref{fig:simul_eval}a left-column)}. Let $N(0,\eta^2)$ denote a normal random variable with mean $0$ and standard deviation $\eta$. We define $X_{tv}$ as a linear Gaussian VAR for $v=1,\ldots,4$ and $t=1,2,\ldots,1000$, whose true CFC has the edges $1\rightarrow 3,2\rightarrow 3, 3\rightarrow 4$. Let $X_{0v}\sim N(0,\eta^2)$ for $v=1,\ldots,4$ and for $t\geq 1$,
    \begin{align*}
    &X_{t1}=1+\epsilon_{t1},~&&X_{t2}=-1+\epsilon_{t2},\\
    &X_{t3}=2X_{(t-1)1}+X_{(t-1)2}+\epsilon_{t3},~&&X_{t4}=2X_{(t-1)3}+\epsilon_{t4}.
    \end{align*}
    where $\epsilon_t\sim N(0,\eta^2)$. It follows that the Rolled Markov Graph with respect to $\bm{\chi}_t$ has edges $1\rightarrow 3, 2\rightarrow 3, 3\rightarrow 4$. We obtain 25 simulations of the entire time series each for different noise levels $\eta \in \{0.1,0.5,1,1.5,2,2.5,3,3.5\}$.
    \vspace{2.5mm}
    \item \underline{Non-linear Non-Gaussian VAR Model (Figure (\ref{fig:simul_eval}a) $2^{\text{nd}}$ left-column)}. Let $U(0,\eta)$ denote a \emph{Uniformly} distributed random variable on the interval $(0,\eta)$. We define $X_{tv}$ as a non-linear non-Gaussian VAR for $v=1,\ldots,4$ and for $t=1,2,\ldots,1000$, whose true CFC has the edges $1\rightarrow 3, 2\rightarrow 3, 3\rightarrow 4$. Let $X_{0v}\sim U(0,\eta)$ for $v=1,\ldots,4$ and for $t\geq 1$, 
    \begin{align*}
    &X_{t1}\sim U(0,\eta),~&&X_{t2}\sim U(0,\eta),\\
    &X_{t3}=4\sin(X_{(t-1)1}+3\cos(X_{(t-1)2})+U(0,\eta),~&&X_{t4}=2\sin(X_{(t-1)3})+U(0,\eta).
    \end{align*}
    The Rolled Markov Graph with respect to $\bm{\chi}_t$ has edges $1\rightarrow 3, 2\rightarrow 3, 3\rightarrow 4$. We obtain 25 simulations of the entire time series each for different noise levels $\eta \in \{0.1,0.5,1,1.5,2,2.5,3,3.5\}$. 
    \vspace{2.5mm}
  \item \underline{Contemporaneous Vector Auto-Regressive Moving Average (VARMA) Model (Figure (\ref{fig:simul_eval}a) $3^{\text{rd}}$ left-column)}
  Let $N(0,\eta)$ denote a normal random variable with mean $0$ and standard deviation $\eta$. We define $X_v(t)$ as a linear Gaussian VAR for $v=1,\ldots,4$ whose true CFC has the edges $1\rightarrow 3,2\rightarrow 3, 3\rightarrow 4$. Let $X_v(0)=N(0,\eta)$ for $v=1,\ldots,4$, and $t=1,2,\ldots,1000$,
    \begin{align*}
    &\bm{X}_t=\begin{pmatrix}
                1\\
                -1\\
                1\\
                2
                \end{pmatrix}+
                \begin{pmatrix}
                    0 & 0 & 0 & 0\\
                    0 & 0 & 0 & 0\\
                    2 & 1 & 0 & 0\\
                    0 & 0 & 2 & 0
                \end{pmatrix} \bm{X}_{t-1} + 
                \begin{pmatrix}
                    0 & 0 & 0 & 0\\
                    0 & 0 & 0 & 0\\
                    0 & 0 & 0 & 0\\
                    2 & 1 & 0 & 0
                \end{pmatrix} \bm{\epsilon}_{t-1}+
                \begin{pmatrix}
                    1 & 0 & 0 & 0\\
                    0 & 1 & 0 & 0\\
                    2 & 1 & 1 & 0\\
                    2 & 1 & 1 & 1
                \end{pmatrix} \bm{\epsilon}_{t}
    \end{align*}
    It follows that the Rolled Markov Graph with respect to $\bm{\chi}_t$ has the edges $1\rightarrow 3, 2\rightarrow 3, 3\rightarrow 4$. Furthermore, $\bm{X}_t$ is faithful with respect to the same graph. We obtain 25 simulations of the entire time series each for different noise levels $\eta \in \{0.1,0.5,1,1.5,2,2.5,3,3.5\}$.
    \vspace{2.5mm}
  \item \underline{Continuous Time Recurrent Neural Network (CTRNN) (Figure (\ref{fig:simul_eval}a) right-column)}. We simulate neural dynamics by Continuous Time Recurrent Neural Networks (\ref{ctrnn}). $u_j(t)$ is the instantaneous firing rate at time $t$ for a post-synaptic neuron $j$, $w_{ij}$ is the linear coefficient to pre-synaptic neuron $i$'s input on the post-synaptic neuron $j$, $I_j(t)$ is the input current on neuron $j$ at time $t$, $\tau_j$ is the time constant of the post-synaptic neuron $j$, with $i,j$ being indices for neurons with $m$ being the total number of neurons. Such a model is typically used to simulate neurons as firing rate units,
    \begin{equation}\label{ctrnn}
    \tau_j \frac{du_j(t)}{dt}=-u_j (t) + \sum_{i=1}^m w_{ij} \sigma (u_i (t)) + I_j (t), \quad j=1,\ldots, m.
    \end{equation}
    We consider a motif consisting of $4$ neurons with $w_{13}=w_{23}=w_{34}=10$ and $w_{ij}=0$ otherwise. We also note that in Eq. \ref{ctrnn}, activity of each neuron $u_j(t)$ depends on its own past. Therefore, the Rolled Markov Graph with respect to $\bm{\chi}_t$ has the edges $1\rightarrow 3,2\rightarrow 3,3\rightarrow 4, 1\rightarrow 1, 2\rightarrow 2, 3\rightarrow 3, 4\rightarrow 4$. The time constant $\tau_i$ is set to 10 msecs for each neuron $i$. We consider $I_i(t)$ to be distributed as an independent Gaussian process with mean 1 and the standard deviation $\eta$. The signals are sampled at a time gap of $e \approx 2.72$ msecs for a total duration of $1000$ msecs.     We obtain 25 simulations of the entire time series each for different noise levels $\eta \in \{0.1,0.5,1,1.5,2,2.5,3,3.5\}$.

\end{enumerate}

The GC graph is computed using the \emph{Nitime} Python library, which fits an MVAR model followed by using the \emph{GrangerAnalyzer} to compute the Granger Causality \citep{rokem2009nitime}. For PC, TPCS and TPCNS, the computation is done using the \emph{TimeAwarePC} Python library \citep{biswas2022statistical2}. The TPCS and TPCNS algorithms estimate the Rolled Markov Graph from the signals with $\tau=1, r=2\tau$.

The choice of thresholds tunes the decision whether a connection exists in the estimate. For PC, TPCS and TPCNS, increasing $\alpha$ in conditional dependence tests increases the rate of detecting edges, but also increases the rate of detecting false positives. We consider $\alpha = 0.01, 0.05, 0.1$ for PC, TPCS and TPCNS. For GC, a likelihood ratio statistic $L_{uv}$ is obtained for testing $A_{uv}(k) = 0$ for $k=1,\ldots,K$. An edge $u\rightarrow v$  is outputted if $L_{uv}$ has a value greater than a threshold. We use a percentile-based threshold, and output an edge $u\rightarrow v$ if $L_{uv}$ is greater than the $100(1-\alpha)^{\text{th}}$ percentile of $L_{ij}$’s over all pairs of neurons $(i,j)$ in the graph \citep{schmidt2016multivariate}. We consider $\alpha = 0.01, 0.05, 0.1$ which corresponds to percentile thresholds of $99\%, 95\%, 90\%$. TPCNS is conducted with 50 subsamples with window length of $50$ msec and frequency cutoff for edges to be equal to $40\%$.

\section{Benchmark Datasets}\label{data:bm}
We use the River Runoff benchmark real dataset from \emph{Causeme} \citep{runge2019inferring,bussmann2021neural}. This is a real dataset that consists of time series of river runoff at different stations. This time series has a daily time resolution and only includes summer months (June-August). The physical time delay of interaction (as inferred from the river velocity) are roughly below one day, hence the dataset has contemporaneous time interactions. This dataset has 12 variables and 4600 time recordings for each variable.

The PC algorithm was implemented with p-value $0.1$ for kernel-based non-linear conditional dependence tests. In river-runoff data, the TPCS algorithm was implemented with $\alpha = 0.05$, and the TPCNS algorithm was implemented with $\tau=2$ and $4$ recordings respectively, as per specification in the datasets, and $\alpha = 0.05$ for the conditional dependence tests. TPCNS was conducted with $50$ subsamples with window length of $50$ recordings and frequency cut-off for edges to be equal to $0.1$.

\end{document}